\documentclass[%
 reprint,
 superscriptaddress,
 amsmath,%amssymb,
 aps,
 prl,
 longbibliography
]{revtex4-2}

\usepackage{graphicx, color}% Include figure files
\usepackage{amsthm}
\usepackage{newtxtext, newtxmath}
\usepackage{bm}% bold math
\usepackage{braket, mathtools}
\usepackage{algorithm}
\usepackage{algpseudocode}
\usepackage{enumitem}
\usepackage[
 colorlinks=true,
 citecolor=blue,
 linkcolor=blue
]{hyperref}

\graphicspath{{./images/}}

\DeclareMathOperator{\argmax}{arg\,max}
\DeclareMathOperator{\argmin}{arg\,min}

\let\tilde\widetilde
\newcommand{\ii}{\mathrm{i}}
\newcommand{\ee}{\mathrm{e}}
\newcommand{\relmiddle}[1]{\mathrel{}\middle#1\mathrel{}}
\theoremstyle{plain}
\newtheorem{theorem}{Theorem}
\newtheorem{corollary}{Corollary}
\newtheorem{lemma}{Lemma}
\newtheorem{proposition}{Proposition}
\newtheorem{propositionalt}{Proposition}
\newenvironment{propositionp}[1]{
  \renewcommand\thepropositionalt{#1}
  \propositionalt
}{\endpropositionalt}
\theoremstyle{definition}
\newtheorem{definition}{Definition}
\theoremstyle{remark}
\newtheorem*{remark}{Remark}

\begin{document}

\title{Anisotropy-induced spin parity effects}
%\thanks{A footnote to the article title}

\author{Shuntaro Sumita}
\email[]{s-sumita@g.ecc.u-tokyo.ac.jp}
\affiliation{%
 Department of Basic Science, The University of Tokyo, Meguro, Tokyo 153-8902, Japan
}%
\affiliation{%
 Komaba Institute for Science, The University of Tokyo, Meguro, Tokyo 153-8902, Japan
}%
\affiliation{%
 Condensed Matter Theory Laboratory, RIKEN CPR, Wako, Saitama 351-0198, Japan
}%

\author{Akihiro Tanaka}
% \email[]{TANAKA.Akihiro@nims.go.jp}
\affiliation{%
Research Center for Materials Nanoarchitectonics, National Institute for Materials Science, Tsukuba, Ibaraki 305-0044, Japan
}%

\author{Yusuke Kato}
% \email[]{yusuke@phys.c.u-tokyo.ac.jp}
\affiliation{%
 Department of Basic Science, The University of Tokyo, Meguro, Tokyo 153-8902, Japan
}%
\affiliation{%
 Quantum Research Center for Chirality, Institute for Molecular Science, Okazaki, Aichi 444-8585, Japan
}%

\date{\today}

\begin{abstract}
 Spin parity effects refer to those special situations where a dichotomy in the physical behavior of a system arises, solely depending on whether the relevant spin quantum number is integral or half-odd integral.
 As is the case with the Haldane conjecture in antiferromagnetic spin chains, their pursuit often derives deep insights and invokes new developments in quantum condensed matter physics.
 Here, we put forth a simple and general scheme for generating such effects in any spatial dimension through the use of anisotropic interactions, and a setup within reasonable reach of state-of-the-art cold-atom implementations.
 We demonstrate its utility through a detailed analysis of the magnetization behavior of a specific one-dimensional spin chain model, an anisotropic antiferromagnet in a transverse magnetic field, unraveling along the way the quantum origin of finite-size effects observed in the magnetization curve that had previously been noted but not clearly understood.
\end{abstract}

\maketitle

%%%%% Introduction %%%%%
\textit{Introduction.}%
It often happens that a pivotal development in quantum magnetism is triggered by the discovery of a \textit{spin parity effect} (SPE), a phenomenon in which the behavior of a magnetic system sharply depends on the \textit{parity of twice the spin quantum number $S$}.
The Haldane conjecture on antiferromagnetic spin chains~\cite{Haldane1981_preprint, Haldane1983, Haldane2017_review}, the prime example of an SPE, asserts that a spectral gap exists between the ground state (GS) and excited states for integer $S$, whereas the corresponding spectrum is gapless for half-odd-integer $S$.
While this claim is now long established, the activity that ensued has since evolved into themes central to present day condensed matter physics.
An example with far-reaching consequences to quantum many-body systems is the no-go theorem of Lieb, Schultz, and Mattis (LSM), which in its original form prohibits the existence of a unique and featureless gapped GS in an $S = 1/2$ Heisenberg chain~\cite{Lieb1961}.
Among its extensions are those to general $S$~\cite{Kolb1985, Affleck1986, Oshikawa1997, Yamanaka1997, Koma2000, Tasaki2018, Tasaki2022, Ogata2021}, higher dimensions~\cite{Oshikawa2000, Hastings2004, Hastings2005, Hastings2010_arXiv, Nachtergaele2007, Bachmann2020, Yao2020, Yao2021, Yao2022_PRL, Yao2022_PRB, Tada2021}, various symmetries~\cite{Parameswaran2013, Watanabe2015, Watanabe2018, Po2017, Yao2020, Yao2021, Yao2022_PRL, Yao2022_PRB, Yao2023_arXiv}, and electron systems~\cite{Parameswaran2013, Watanabe2015, Watanabe2018, Lu2020}.
The quantum dynamics of solitons~\cite{Braun1996_PRB, Braun1996_JAP, Kodama2023} and skyrmions~\cite{Takashima2016} in chiral magnets hosting Dzyaloshinskii--Moriya interactions is another active research front where SPEs have recently been identified;
there, the soliton/skyrmion states were found to have spin-parity-dependent crystal momenta.

Given how SPEs continue to shed new light on quantum condensed matter, it is desirable to have a generic and comprehensive scheme with which to generate them.
The purpose of this Letter is to put forth just such a method.
Our approach incorporates \textit{anisotropic interactions} as its key element and works in any spatial dimension, which is to be compared with how SPEs, including those mentioned above, usually have dimension-specific origins.
We are also motivated by rapid theoretical~\cite{Altman2003} and experimental~\cite{Chung2021, deHond2022} progress in cold-atom physics that have come a long way toward implementing higher-$S$ quantum spin systems with strong anisotropy.

To best illustrate our strategy, we apply it to a one-dimensional (1D) quantum spin system which has the merit of (1) being amenable to detailed analysis and (2) exhibiting a clear SPE that manifests itself in raw finite-size numerical data.
Feasibility aside, this problem turns out to be interesting in its own rights:
The finite-size effect studied, while long known, has a topological significance (in the sense of Haldane~\cite{Haldane1981_JPhys}) that had gone unnoticed.
The anisotropy-induced SPE in this model is nontrivial in that it evades detection by LSM-type arguments.
Finally, it can be considered an immediate target for cold-atom implementations.
Generalizations to a far wider range of quantum magnets will be discussed afterwards.

%%%%% Model and Numerical results %%%%%
\textit{Model and exact diagonalization.}%
The Hamiltonian of our choice is
\begin{equation}
 \hat{\mathcal{H}} = J \sum_{j=1}^{L} \hat{\bm{S}}_j \cdot \hat{\bm{S}}_{j+1} - H \sum_{j=1}^{L} \hat{S}_j^z + K \sum_{j=1}^{L} (\hat{S}_j^y)^2.
 \label{eq:Hamiltonian}
\end{equation}
The $J (> 0)$ and $H$ terms are the exchange and Zeeman interactions, respectively, while the $K (\geq 0)$ term is an easy-plane single-ion anisotropy.
$L$ stands for the number of sites.
We impose a periodic boundary condition, with the system always consisting of an even number of sites.

\begin{figure*}
 \includegraphics[width=\linewidth, pagebox=artbox]{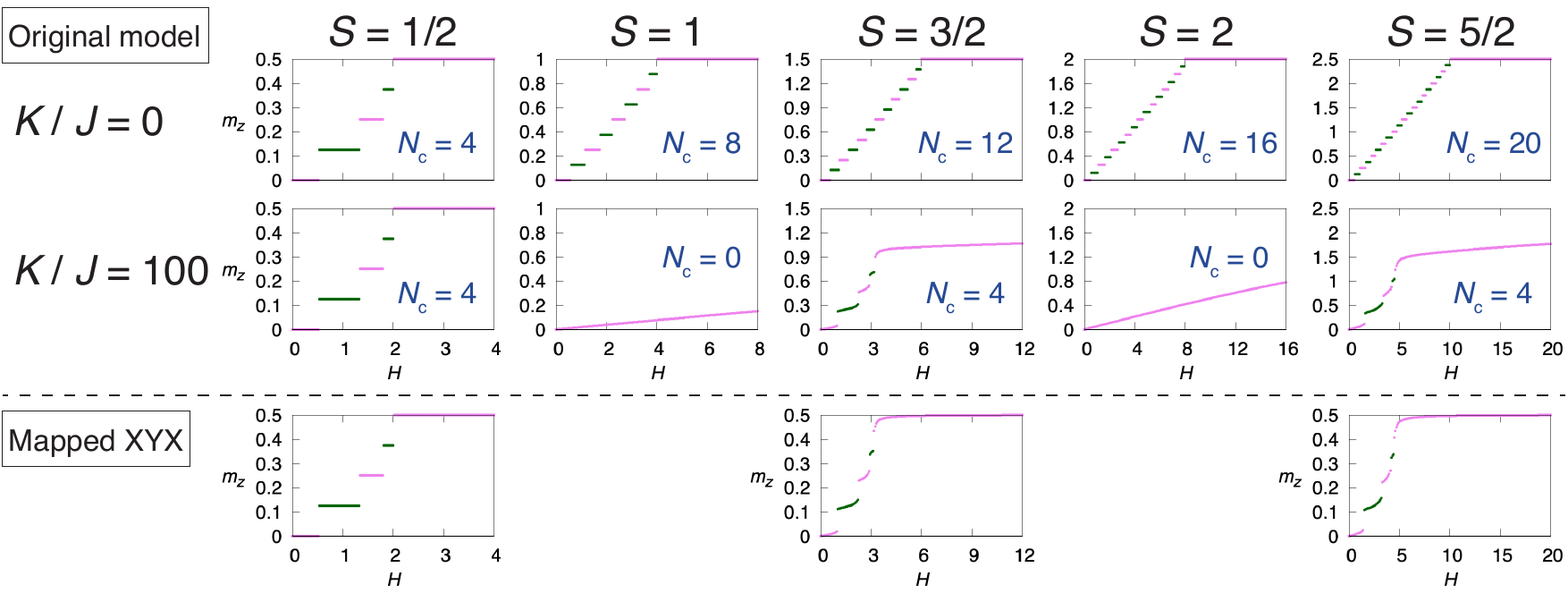}
 \caption{Magnetization $m_z$ vs magnetic field $H$ curves obtained by numerical calculations. $L$ is set to $8$ for all panels. The violet and green lines represent that the GS has crystal momenta $0$ and $\pi$, respectively. Top and middle panels: Results for zero and large anisotropy in the original model. Bottom panels: Results in the XYX model that are obtained by the mapping for large $K$ and half-odd-integer spins.}
 \label{fig:numerical_results}
\end{figure*}

We display in the first two rows of Fig.~\ref{fig:numerical_results} how the GS expectation value of the magnetization, $m_z := \frac{1}{L} \sum_{j=1}^{L} \braket{\hat{S}_j^z}_{\mathrm{GS}}$, evolves as a function of $H$.
Here, $J$ is set to unity.
The numerical exact diagonalization was performed for various values of $S$ with the use of QuSpin~\cite{Weinberg2017, Weinberg2019, Note1}.
\footnotetext{The source code used for exact diagonalization is available in \href{https://github.com/shuntarosumita/antiferro_spin_parity/tree/main/ed}{https://github.com/shuntarosumita/antiferro\_spin\_parity/tree/main/ed}.}%
Numerics for other parameter choices are given in the Supplemental Material (SM)~\cite{Note2}.
\footnotetext{See Supplemental Material, which include Refs.~\cite{Catalano2022, Schollwock1995, Tonegawa2011, Kurmann1982, Muller1985, Dmitriev2002, ITensor, ITensor-r0.3, Pollmann2010, Gioia2022, TanakaTotsukaHu2009, fradkin_book_2013, Wu_Yang_1976, Takayoshi_Totsuka_Tanaka2015, Haldane_PRL_1988, Sachdev2002, Affleck_merons_PRL_1986, Gu_Wen_PRB_2009, Fuji_PRL_2009, Sachdev_2023, Wen_2004} where calculational details and complete proofs, as well as a semiclassical account of the problem are provided.}%
The jumps in the magnetization are due to level crossings (LCs) between the GS and the first excited state; such features had been noted~\cite{Parkinson1985, Takahashi1991, Cabra1997, Arlego2003} but not fully understood.
We find that they are accompanied by the alternation in the GS's crystal momentum between two values, $0$ (violet) and $\pi$ (green).
Furthermore, a marked difference in behavior was found depending on the magnitude of $K$.
This is summarized in the tabular information below, where we indicate by $N_{\mathrm{c}}$ the number of LCs:
\begin{equation}
 \begin{array}{ccc}
  N_{\mathrm{c}} & \text{odd} \ 2S & \text{even} \ 2S \\ \hline
  K = 0, \text{small} \ K & LS & LS \\
  K \gg J & L/2 & 0 \\
 \end{array}
 \label{eq:level_crossings}
\end{equation}
When $K$ is small or zero the jumps are present irrespective of the spin parity, while in the strongly anisotropic regime $K \gg J$ they manifest themselves only when $S$ is a half-odd integer.
We discuss the two cases in turn.

%%%%% Tomonaga--Luttinger liquid %%%%%
\textit{$K = 0$\textup{:} XX model and Tomonaga--Luttinger liquid (TLL).}%
Our objective here is to explain the LCs with $N_{\mathrm{c}} = LS$ between the $0$- and $\pi$-momentum states for $K = 0$ (top panels in Fig.~\ref{fig:numerical_results}).
Before dealing with the full model Eq.~\eqref{eq:Hamiltonian}, it is instructive to warm up with the spin-$1/2$ XX model, $\hat{\mathcal{H}}_{\mathrm{XX}} = J \sum_{j=1}^{L} ( \hat{S}_j^x \hat{S}_{j+1}^x + \hat{S}_j^y \hat{S}_{j+1}^y ) - H \sum_{j=1}^{L} \hat{S}_j^z$, which allows for an intuitive understanding of the momentum switching.
The Jordan--Wigner (JW) transformation~\cite{Jordan-Wigner} maps this model into a noninteracting spinless fermion,
\begin{equation}
 \hat{\mathcal{H}}_{\mathrm{XX}} = \sum_{l=1}^{L} (J \cos k_l - H) \hat{a}_{k_l}^\dagger \hat{a}_{k_l} + \frac{LH}{2},
 \label{eq:XX_Hamiltonian_JW}
\end{equation}
where $\hat{a}_{k_l}$ is the annihilation operator of the JW fermion carrying momentum $k_l$.
The Hamiltonian commutes with the fermion number operator $\hat{N}_a := \sum_{l=1}^{L} \hat{a}_{k_l}^\dagger \hat{a}_{k_l}$, where the sum is taken over $k_l = (2l - 1) \pi / L$ when $N_a$ is even, while $k_l = 2 (l - 1) \pi / L$ for odd $N_a$.
Block-diagonalizing $\hat{\mathcal{H}}_{\mathrm{XX}}$ within eigensectors of $\hat{N}_a$, we show in Fig.~\ref{fig:XX_JW}(a) the energy eigenvalues $E_{N_a}(H)$ as a function of $H$, where plots for even (odd) $N_a$ are colored in violet (green).
As the magnetic field $H$ is ramped up the GS eigenvalue of $\hat{N}_a$ increases by unit increments which translates back to the number of up spins in the original model, $\hat{N}_a = \hat{S}_{\mathrm{tot}}^z + L/2$, where $\hat{S}_{\mathrm{tot}}^z = \sum_{j=1}^{L} \hat{S}_j^z$.
We depict the cosine term in Eq.~\eqref{eq:XX_Hamiltonian_JW} for $N_a = 4$ and $5$ in Figs.~\ref{fig:XX_JW}(b) and \ref{fig:XX_JW}(c), respectively.
Clearly the total momentum $k_{\mathrm{tot}} := \Braket{\sum_{l=1}^{L} k_l \hat{a}_{k_l}^\dagger \hat{a}_{k_l}} \pmod{2\pi}$ is zero for even $N_a$, and $\pi$ for odd $N_a$.
Returning now to the Heisenberg model Eq.~\eqref{eq:Hamiltonian} with $K = 0$, this momentum-counting argument no longer applies as the $\sum_{j=1}^{L} J \hat{S}_j^z \hat{S}_{j+1}^z$ term generates interactions between the JW fermions.
Using the Bethe ansatz, however, the total momentum $k_{\mathrm{tot}}$ can be shown to remain unaffected, despite the phase shift each momentum $k_l$ receives from particle scattering~\cite{Giamarchi_textbook, Note2}.
Whether the GS momentum is $0$ or $\pi$ is thus determined solely by the parity of the number of up spins.

\begin{figure}
 \includegraphics[width=\linewidth, pagebox=artbox]{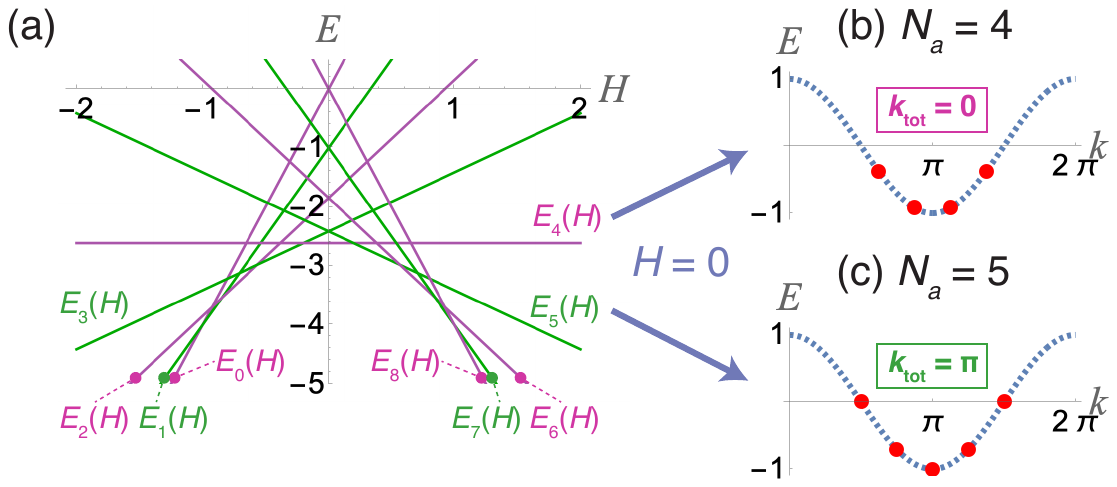}
 \caption{(a) Energy eigenvalues of the spin-$1/2$ XX chain with $L = 8$ for each $\hat{N}_a$ sector. Momentum--energy relations for $H = 0$ in the sectors with $N_a = 4$ and $5$ are shown in (b) and (c), respectively.}
 \label{fig:XX_JW}
\end{figure}

The preceding argument was limited to $S = 1/2$.
To generalize to arbitrary spin, we incorporate the powerful machinery of the TLL theory, known to correctly capture (for higher $S$ cases as well) the behavior of the gapless, linear dispersion of the $K = 0$ model, i.e., the Heisenberg model under an applied magnetic field.
In this framework, a low-energy excited state is fully characterized by a quartet of integer-valued quantum numbers: $\Delta N$ (charge excitation), $\Delta D$ (current excitation), and $N^\pm$ (right/left-moving particle-hole pair).
The corresponding energy eigenvalue and momentum are given by~\cite{Kawakami1992}
\begin{align}
 \Delta E &= \frac{2\pi v}{L} \left[ \frac{(\Delta N)^2}{4K_{\mathrm{L}} }+ K_{\mathrm{L}} (\Delta D)^2 + N^+ + N^- \right] + \mu \Delta N,
 \label{eq:TLL_excitation_e} \\
 \Delta k &= 2\pi\rho \Delta D + \frac{2\pi}{L} [\Delta N \Delta D + N^+ - N^-] + \pi \Delta N,
 \label{eq:TLL_excitation_p}
\end{align}
where $K_{\mathrm{L}}$, $v$, $\mu$, and $2\pi\rho$ are each the Luttinger liquid parameter, the excitation's velocity, the chemical potential, and the intrinsic momentum~\cite{Haldane1982_PRL, *[][{ (erratum).}]{Haldane1982_PRL_erratum}, Konik2002}.
The last term in Eq.~\eqref{eq:TLL_excitation_p} arises from antiferromagnetic correlations inherent to our lattice model; retaining it on taking the continuum limit is crucial.
The only quantum number coupling to a magnetic field is $\Delta N$; it then follows that the state $(\Delta N, \Delta D, N^+, N^-) = (1, 0, 0, 0)$, which amounts to a change in momenta by $\Delta k = \pi$, always becomes the next GS as the field is increased.
The quantity $\braket{\hat{S}_{\mathrm{tot}}^z}_{\mathrm{GS}}$, which is conserved when $K=0$, increases by one each time a crossover with $\Delta N = 1$ takes place.
Accordingly the value of $\braket{\hat{S}_{\mathrm{tot}}^z}_{\mathrm{GS}}$ undergoes the sequence: $0 \rightarrow 1 \rightarrow \dots \rightarrow LS$.
This is in agreement with the relation $N_{\mathrm{c}} = LS$ [Eq.~\eqref{eq:level_crossings}], which we extracted from the magnetization processes in the top panels of Fig.~\ref{fig:numerical_results}.

\begin{figure}
 \includegraphics[width=\linewidth, pagebox=artbox]{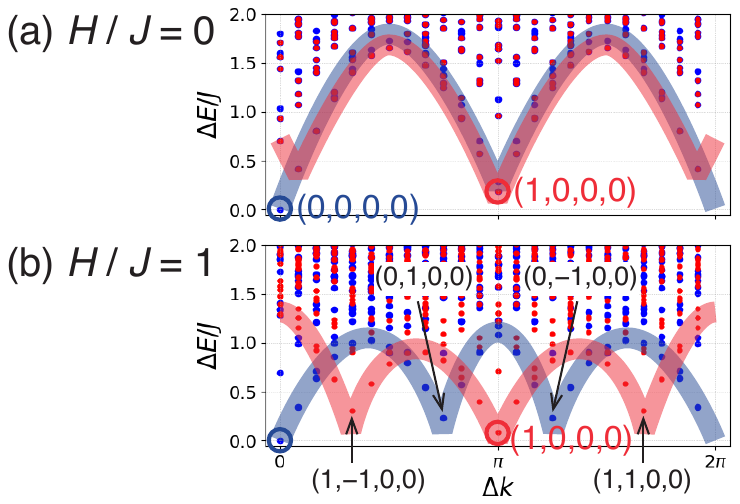}
 \caption{Energy dispersions of the spin-$1/2$ Heisenberg chain with $L = 24$ for (a) $H / J = 0$ and (b) $1$. In both panels, the GS is positioned at $(\Delta E, \Delta k) = (0, 0)$, and the other points show the spectrum of excited states. Each blue (red) point represents the spectrum in the $\Delta S_{\mathrm{tot}}^z = \Delta N = 0$ ($1$) sector. Sets of quantum numbers $(\Delta N, \Delta D, N^+, N^-)$ are shown for some characteristic points. The translucent lines are guides for the eyes.}
 \label{fig:Heisenberg_TLL}
\end{figure}

Figures~\ref{fig:Heisenberg_TLL}(a) and \ref{fig:Heisenberg_TLL}(b) show in full detail the energy dispersions of the spin-$1/2$ Heisenberg model, which we use to check the validity of the above picture.
States for which $\Delta S_{\mathrm{tot}}^z := \braket{\hat{S}_{\mathrm{tot}}^z} - \braket{\hat{S}_{\mathrm{tot}}^z}_{\mathrm{GS}} (= \Delta N)$ is $0$ and $1$ are each plotted as blue and red points.
The low-energy spectrum contains conformal towers, located at $k = 0, \pm 2\pi\rho$ for $\Delta S_{\mathrm{tot}}^z = 0$ (blue lines), and at $k = \pi, \pi \pm (2\pi\rho + 2\pi/L)$ for $\Delta S_{\mathrm{tot}}^z = 1$ (red lines).
The results are consistent with Eqs.~\eqref{eq:TLL_excitation_e} and \eqref{eq:TLL_excitation_p}; in particular, $\rho$ is equal to $1/2$ for the ``half-filling'' case $H = 0$ [cf. Fig.~\ref{fig:XX_JW}(b)].

%%%%% Z2 symmetry and crystal momentum %%%%%
\textit{$\mathbb{Z}_2$ symmetry and crystal momentum.}%
We proceed to nonzero $K$ cases, where $\hat{S}_{\mathrm{tot}}^z$ is not conserved.
We can still take advantage of a discrete $\mathbb{Z}_2 \times \mathbb{Z}_2^T$ symmetry of the Hamiltonian of Eq.~\eqref{eq:Hamiltonian}, where the unitary part is generated by $\hat{Z} := \bigotimes_{j=1}^{L} \ee^{\ii \pi (S - \hat{S}_j^z)}$, i.e., a $\pi$-rotation of all spins with respect to the $z$ axis.
Noting that $\hat{Z}^2 = 1$ for any $S$, the Hamiltonian can be block-diagonalized into sectors labeled by the $\mathbb{Z}_2$ values $\hat{Z} = \pm 1$: $\hat{\mathcal{H}} = \hat{\mathcal{H}}^{(+)} \oplus \hat{\mathcal{H}}^{(-)}$.
The following is true for arbitrary positive values of $K$:
\begin{theorem}
 \label{thm:Z2_momentum}
 When the GS of the Hamiltonian Eq.~\eqref{eq:Hamiltonian} is a simultaneous eigenstate with $\hat{Z} = +1 \ (-1)$, it has a crystal momentum $0 \ (\pi)$.
\end{theorem}
A sketch of the proof goes as follows;
the full details are given in the SM~\cite{Note2}.
Let $\ket{\bm{m}} := \ket{m_1 m_2 \dots m_L} \ (m_j = -S, -S+1, \dots, S)$ be the usual spin basis such that $\hat{S}_j^z \ket{\bm{m}} = m_j \ket{\bm{m}}$.
Introducing a \textit{signed basis} $\ket{\widetilde{\bm{m}}} := (-1)^{\delta(\bm{m})} \ket{\bm{m}}$ with $\delta(\bm{m}) := \sum_{j=1}^{L} j (S - m_j)$, one can show that (i) off-diagonal elements of the Hamiltonian are nonpositive, and (ii) both eigensectors with $\hat{Z} = \pm 1$ are irreducible.
From (i) and (ii), the Perron--Frobenius theorem~\cite{Perron1907, Frobenius1912} applied to the Hamiltonian block $\hat{\mathcal{H}}^{(\pm)}$ leads to the uniqueness of the GS in each eigensector given by $\Ket{\Psi_{\mathrm{GS}}^{(\pm)}} = \sum_{\bm{m} \in V^{(\pm)}} a(\bm{m}) \ket{\widetilde{\bm{m}}}$, where $a(\bm{m}) > 0$ and $V^{(\pm)} := \left\{\bm{m} \relmiddle| \hat{Z} \ket{\widetilde{\bm{m}}} = \pm \ket{\widetilde{\bm{m}}}\right\}$.
Further, the one-site translation operator $\hat{T}$ affects the sign of the basis: $(-1)^{\delta(\hat{T}(\bm{m}))} = (-1)^{\sum_{j} (S - m_j)} (-1)^{\delta(\bm{m})} = \hat{Z} (-1)^{\delta(\bm{m})}$.
Combining these, a little algebra shows that $\hat{T} \Ket{\Psi_{\mathrm{GS}}^{(\pm)}} = \pm \Ket{\Psi_{\mathrm{GS}}^{(\pm)}}$.

Theorem~\ref{thm:Z2_momentum} implies that the LCs between $0$- and $\pi$-momentum states are characterized by the parity switching of $\hat{S}_{\mathrm{tot}}^z$, although $\hat{S}_{\mathrm{tot}}^z$ itself is in general not conserved.
Moreover, as long as $\hat{Z}$ is preserved, the LCs can survive even when the translation symmetry is broken.

%%%%% Perturbation theory %%%%%
\textit{Large $K$\textup{:} perturbation theory.}%
We now address the SPE arising at large $K$ (middle panels in Fig.~\ref{fig:numerical_results}).
This can be understood in terms of a perturbation theory applicable for $K \gg J, H$.
The nonperturbative Hamiltonian is just $\hat{\mathcal{H}}_0 = K \sum_{j=1}^{L} (\hat{S}_j^y)^2$, which reduces to a single-site problem:
Remembering that $K$ is positive, the GS of $\hat{\mathcal{H}}_0$ is %clearly
just a product state of doublets $\ket{S_j^y = \pm \frac{1}{2}}$ residing on each site when $S$ is a half-odd integer.
Meanwhile a unique GS $\ket{S_j^y = 0}$ is formed for integer $S$ (Fig.~\ref{fig:K_model_energy}).
Since the magnitude of the energy gap from the GS is of order $O(K)$, excited states are negligible in discussing the large-$K$ physics.
In other words, at low energies our model in the large-$K$ limit translates into effective spin-$0$ and spin-$1/2$ systems for integer and half-odd-integer $S$ cases, respectively~\cite{Note3}.
\footnotetext{In Ref.~\cite{Oshikawa1997} this observation was applied to the magnetization plateau problem for an $S = 3/2$ spin chain in a longitudinal magnetic field.}%
In the former case, the GS is unique and trivial even when the perturbation $\hat{\mathcal{H}} - \hat{\mathcal{H}}_0$ is switched on, as long as $K \gg J, H$ is satisfied, explaining the absence of LCs.
This GS can be identified with the so-called large-$D$ phase, a gapped phase known to be topologically trivial~\cite{Schulz1986, denNijs1989, Chen2003, Pollmann2012, Huang2021}.

\begin{figure}
 \includegraphics[width=.9\linewidth, pagebox=artbox]{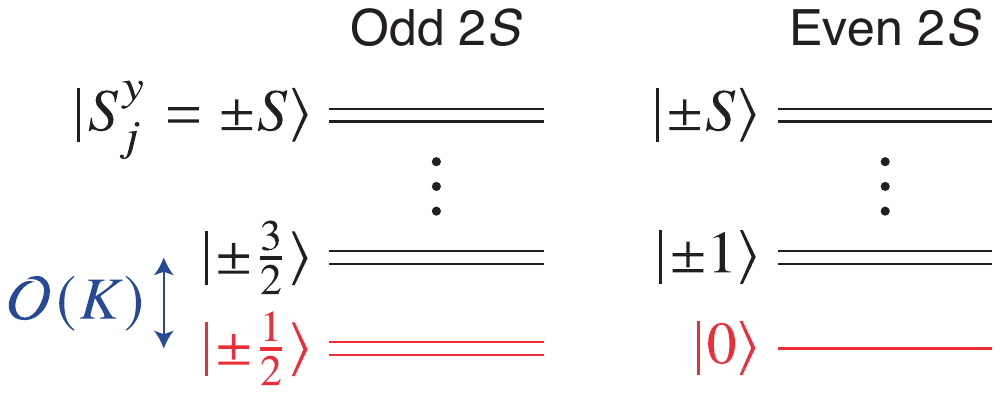}
 \caption{Energy levels of the $K$ model on the $j$th site. The GS degeneracy is two (one) in odd-$2S$ (even-$2S$) systems.}
 \label{fig:K_model_energy}
\end{figure}

The situation is quite different in half-odd-integer spin cases owing to the GS degeneracy of the nonperturbative Hamiltonian.
Working within a first-order perturbation scheme, we derive the following effective spin-$1/2$ XYX model,
\begin{equation}
 \hat{\mathcal{H}}_{\mathrm{map}} = \tilde{J} \sum_{j=1}^{L} (\hat{s}_j^x \hat{s}_{j+1}^x + \Delta \hat{s}_j^y \hat{s}_{j+1}^y + \hat{s}_j^z \hat{s}_{j+1}^z) - \tilde{H} \sum_{j=1}^{L} \hat{s}_j^z,
 \label{eq:Hamiltonian_map}
\end{equation}
where $\hat{\bm{s}}_j$ is the effective spin-$1/2$ operator on the $j$th site; see the SM for derivations~\cite{Note2}.
The coefficients are defined as $\tilde{J} := J / \Delta$, $\tilde{H} := H / \sqrt{\Delta}$ with $\Delta := (S + 1/2)^{-2}$.
Notice that the single-ion anisotropy of the original Hamiltonian has effectively transformed here into an anisotropy of the exchange interaction.
Numerical calculations for $\hat{\mathcal{H}}_{\mathrm{map}}$ are shown in the bottom panels of Fig.~\ref{fig:numerical_results}, where we find a good agreement between the obtained magnetization curve in the low-field regime with those of the original model (middle panels of Fig.~\ref{fig:numerical_results}).
This reduction to effective spin-$1/2$ systems naturally explains the entry $N_{\mathrm{c}} = L/2$ in Eq.~\eqref{eq:level_crossings}.

%%%%% Generalization to higher dimensions and implementation %%%%%
\textit{Generalization to higher dimensions and implementation.}%
Crucially one notices that anisotropy-induced SPEs are \textit{not} specific to our model, Eq.~\eqref{eq:Hamiltonian}:
In the presence of a dominant $K$ term, the spectrum \textit{always} reduces to Fig.~\ref{fig:K_model_energy} in the large-$K$ limit, whatever the spatial dimension $D$, the lattice geometry, or the Hamiltonian itself are.
Our perturbative scheme applies therefore to a much wider variety of spin systems.

Consider how this works for nearest-neighbor Heisenberg models on $D$-dimensional hypercubic lattices in the absence of external fields.
As in $D = 1$, the spectral gap $\sim O(K)$ remains robust for integer spins, while the low-energy physics of the half-integer spin systems are represented by effective spin-$1/2$ XYX models.
The latter spontaneously breaks the spin rotation symmetry around the $y$ axis, accompanied by a gapless Nambu--Goldstone mode.

Though new features arise, this SPE persists when terms which disrupt spin ordering (e.g., frustrated exchange and four-body interactions) are added.
A trivially gapped state continues to form for integer spins, whereas the effective spin-$1/2$ model for the half-odd-integer spin case can now also sustain a degenerate spin-Peierls state~\cite{Majumdar1969, Haldane1982_PRB, *[][{ (erratum).}]{Haldane1982_PRB_erratum}, Read1989, Read1990, Sachdev2002} as well as a topologically ordered phase (see the SM for digressions~\cite{Note2}).
Somewhat surprisingly, this dichotomy arises even under conditions where the LSM theorem~\cite{Lieb1961, Kolb1985, Affleck1986, Oshikawa1997, Yamanaka1997, Koma2000, Tasaki2018, Tasaki2022, Ogata2021, Oshikawa2000, Hastings2004, Hastings2005, Hastings2010_arXiv, Nachtergaele2007, Bachmann2020, Yao2020, Yao2021, Yao2022_PRL, Yao2022_PRB, Tada2021, Parameswaran2013, Watanabe2015, Watanabe2018, Po2017, Yao2023_arXiv, Lu2020} predicts no SPE.
\textit{The foregoing discussion on Eq.~\eqref{eq:Hamiltonian}, a 1D model with only an on-site $\mathbb{Z}_2 \times \mathbb{Z}_2^T$ symmetry, in fact serves as an example of this highly nontrivial case.}
A dominant $K$ term thus constitutes a simple but powerful point of departure for creating SPEs.

Finally, we comment on the implications for experiments. Higher-$S$ systems of finite size are becoming accessible with recent cold-atom technologies.
Indeed, the 1D Hamiltonian Eq.~\eqref{eq:Hamiltonian}, at least for $S = 1$, has recently been implemented in ultracold ${}^{87}$Rb atoms in an optical lattice, where the single-ion anisotropy is controllable~\cite{Chung2021, deHond2022}.
Although ferromagnetic correlations were realized in these experiments, it is possible to make them antiferromagnetic by using negative temperature states~\cite{Braun2013, Chen2023_coldatom}.
There are also reasons to anticipate the relevance of our approach to magnetic materials as well.
Notably, our numerical results show that SPEs persist even when $K$ and $J$ are of the same order (see Fig.~S1 in SM~\cite{Note2}), a regime better suited for materials search.
Moreover, we expect traces of the SPE to be detectable in the thermodynamic limit, to which actual materials correspond:
Half-odd-integer spin systems would exhibit a small hysteresis in the magnetization curve due to the dense LCs, as observed in a classical chiral magnet~\cite{Tsuruta2016}.
This is not the case with integer spin systems.

%%%%% Summary %%%%%
\textit{Summary.}%
We analyzed an antiferromagnetic spin chain with easy-plane anisotropy, focusing on LCs between states with momenta $0$ and $\pi$ which alternately appear as a transverse magnetic field is increased.
At zero or small anisotropy the LCs were accounted for using TLL theory and exact symmetry arguments.
At strong anisotropy the LCs occur only for half-odd-integer spins.
We showed how the latter is a realization of a very generic anisotropy-induced SPE, whose implementation is likely within reach of cold-atom technologies and magnetic materials science.

%%%%%%%%%% Acknowledgments %%%%%%%%%%
\begin{acknowledgments}
 The authors are grateful to Sohei Kodama for earlier collaboration on related topics, and Shunsuke C. Furuya, Akira Furusaki, and Masaya Kunimi for informative discussions.
 This work was supported by JSPS KAKENHI Grants No.~JP23K13056, No.~JP23K03333, No.~JP20K03855, No.~JP21H01032, and No.~JP19K03662, and JST CREST Grant No.~JPMJCR19T2.
 The computation in this work was done using the facilities of the Supercomputer Center, the Institute for Solid State Physics, the University of Tokyo.
\end{acknowledgments}

\clearpage

\renewcommand{\thesection}{S\arabic{section}}
\renewcommand{\theequation}{S\arabic{equation}}
\setcounter{equation}{0}
\renewcommand{\thefigure}{S\arabic{figure}}
\setcounter{figure}{0}
\renewcommand{\thetable}{S\arabic{table}}
\setcounter{table}{0}
\setcounter{theorem}{0}
\makeatletter
\c@secnumdepth = 2
\makeatother

\onecolumngrid

\begin{center}
 {\large \textmd{Supplemental Materials:} \\[0.3em]
 {\bfseries Anisotropy-Induced Spin Parity Effects}} \\[1em]
 Shuntaro Sumita, Akihiro Tanaka, and Yusuke Kato
\end{center}

\setcounter{page}{1}
\section{Numerical results for other values of $K$ and $L$}
\begin{figure}[b]
 \includegraphics[width=\linewidth, pagebox=artbox]{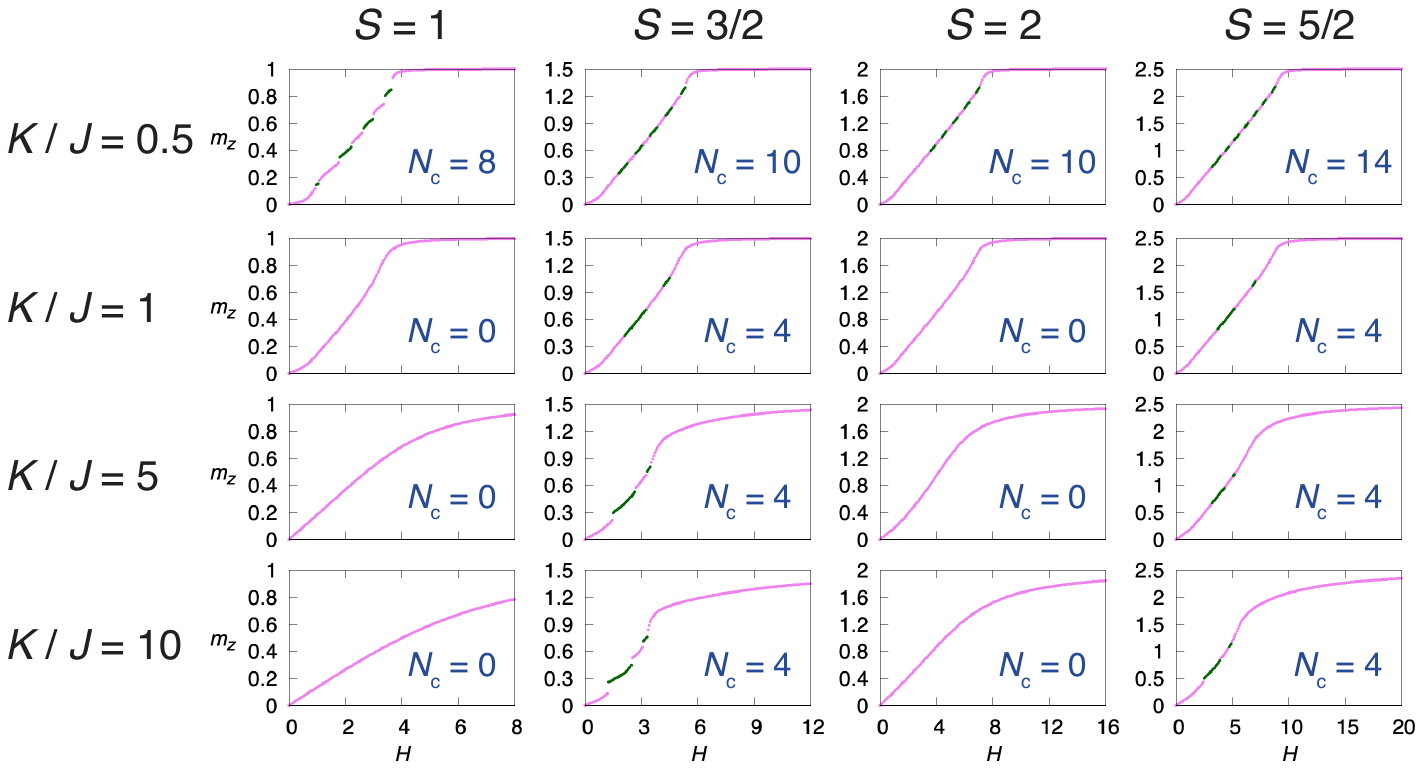}
 \caption{Magnetization $m_z$ vs magnetic field $H$ curves for $K / J = 0.5$, $1$, $5$, and $10$ obtained by exact diagonalization of the Heisenberg model [Eq.~(1) of the main text]. $L$ is set to $8$ for all panels. The violet and green lines represent that the ground state has crystal momenta $0$ and $\pi$, respectively.}
 \label{fig:numerical_results_supple_K}
\end{figure}

In Fig.~1 in the main text, the magnetization curves of a Heisenberg chain with a single-ion anisotropy under a transverse magnetic field for $K/J = 0$ and $100$ are exhibited.
As additional information, in Fig.~\ref{fig:numerical_results_supple_K}, we show numerical results for other anisotropy parameters: $K/J = 0.5$, $1$, $5$, and $10$.
For $K/J \geq 1$, the number of level crossings $N_{\mathrm{c}}$ is $L/2$ and $0$ for half-odd-integer and integer spins, respectively.
Therefore, the anisotropy-induced spin parity effects start to occur when the anisotropy $K$ is comparable to $J$.
With the mapping to the effective model (or an alternative small-parameter expansion) unavailable, it is difficult to extract a general mechanism/reason behind the appearance of the spin parity effect in the intermediate region.
We here thus opt to numerically investigate the critical value of $K$ at zero magnetic field.
According to earlier studies on the $J$--$K$ phase diagrams for $S = 1$ and $2$ Heisenberg chains, the transition from the Haldane/XY to the large-$K$ ($D$) phase occurs at around $K/J \approx 1$ and $3$, respectively~\cite{Chen2003_SM, Schollwock1995_SM, Tonegawa2011_SM}.
When $K/J$ exceeds its value at the transition point, one expects the behavior of the ground state to be qualitatively the same as that in the large-$K$ limit.
Unfortunately, it was difficult to actually check the expected difference in the critical value between $S = 1$ and $2$ in our exact diagonalization, since the spectral gap due to the finite size $L \approx 10$ turns out to be larger than the excitation gap in the thermodynamic limit.

\begin{figure}[tbp]
 \includegraphics[width=\linewidth, pagebox=artbox]{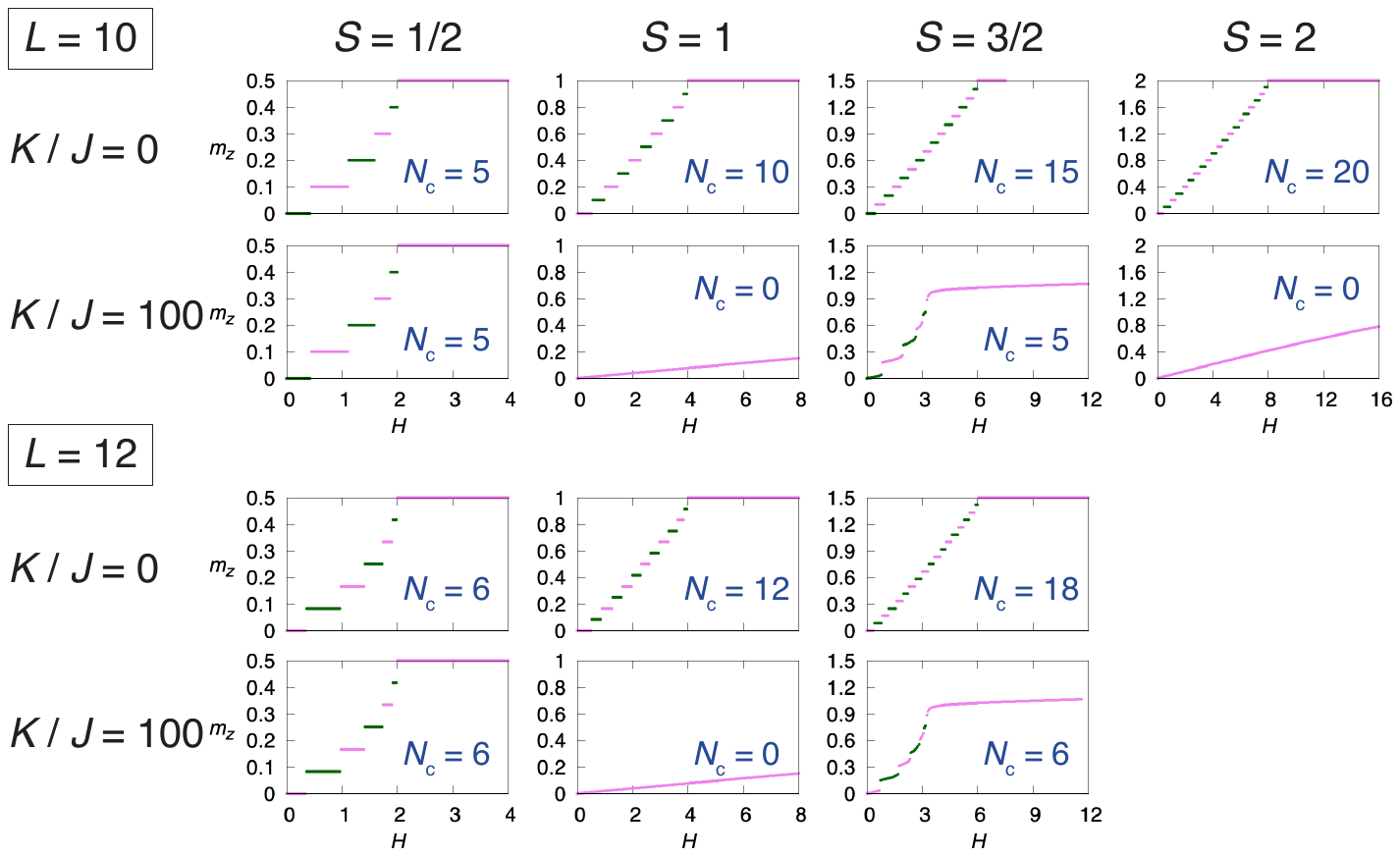}
 \caption{Magnetization $m_z$ vs magnetic field $H$ curves for $L = 10$ and $12$ obtained by exact diagonalization of the Heisenberg model [Eq.~(1) of the main text].}
 \label{fig:numerical_results_supple_L}
\end{figure}

In Fig.~\ref{fig:numerical_results_supple_L}, we also show the system size dependence of the magnetization curve.
For both sizes $L = 10$ and $12$, the number of level crossings $N_{\mathrm{c}}$ satisfies the relations in Eq.~(2) in the main text.
Therefore, our proposal of the anisotropy-induced spin parity effects would be robust against the change in the system size.

While this study focuses solely on the situation where the number of sites $L$ is even, let us briefly remark on odd-$L$ cases.
For example, Catalano \textit{et al.} discussed the behavior of a spin-$1/2$ XYZ chain for odd $L$s and revealed level crossings between different ground-state momenta~\cite{Catalano2022_SM}, which are similar to our results.
The new aspect for odd $L$ is that various finite-momentum ground states appear because the momentum $k = \pi$ is ill-defined.
Therefore, we expect a similar phenomenon to occur in our Heisenberg model.

\section{Total momentum of spin-$1/2$ XXZ model from Bethe ansatz}
We here show that the total momentum of the ground state is always $0$ or $\pi$ in a one-dimensional $S = 1/2$ XXZ model under a longitudinal magnetic field,
\begin{equation}
 \hat{\mathcal{H}}_{\mathrm{XXZ}} = J \sum_{j=1}^{L} (\hat{S}_j^x \hat{S}_{j+1}^x + \hat{S}_j^y \hat{S}_{j+1}^y + \Delta \hat{S}_j^z \hat{S}_{j+1}^z) - H \sum_{j=1}^{L} \hat{S}_j^z,
\end{equation}
where $J > 0$, $|\Delta| \leq 1$, and $L$ is even.
The ground state energy of the Hamiltonian is obtained by using the Bethe ansatz, which is given by summations from free particles~\cite{Giamarchi_textbook_SM},
\begin{equation}
 E = J \sum_{l=1}^{N_a} \cos k_l - H \left(\frac{L}{2} - N_a\right) + J\Delta \left(\frac{L}{4} - N_a\right),
\end{equation}
where $N_a$ represents the number of up spins.
The momentum $k_l$ is given as a solution of the Bethe equation taking the phase shift into account:
\begin{align}
 L k_l &= 2\pi I_l + \sum_{l'=1}^{N_a} \Theta(k_l, k_{l'}) \quad (l = 1, 2, \dots, N_a),
 \label{eq:Bethe_equation} \\
 \Theta(k, k') &= 2 \arctan\left\{ \frac{\Delta \sin[(k - k')/2]}{\Delta \cos[(k - k')/2] - \cos[(k + k')/2]} \right\},
 \label{eq:phase_shift}
\end{align}
where $I_1 < I_2 < \dots < I_{N_a}$ are integers (half-odd integers) when $N_a$ is odd (even).
Assuming that $I_1, \dots, I_{N_a}$ correspond to the equally spaced momenta in the $S = 1/2$ XX model (see the main text), they are given by
\begin{equation}
 I_l = \frac{L}{2} + \frac{2l - N_a - 1}{2} \quad (l = 1, 2, \dots, N_a),
 \label{eq:I_l}
\end{equation}
for both parity of $N_a$.
From Eq.~\eqref{eq:phase_shift}, furthermore, we can easily show that the phase shift satisfies the following equation:
\begin{equation}
 \Theta(k, k') = - \Theta(-k, -k').
\end{equation}
When $k_l$ is a solution of the Bethe equation \eqref{eq:Bethe_equation}, $-k_l$ is also a solution, as a consequence of the inversion symmetry of the Hamiltonian $\hat{H}_{\mathrm{XXZ}}$.
Therefore,
\begin{equation}
 \sum_{l=1}^{N_a} \sum_{l'=1}^{N_a} \Theta(k_l, k_{l'}) = 0.
 \label{eq:Theta_zero}
\end{equation}
Equations~\eqref{eq:I_l} and \eqref{eq:Theta_zero} result in the following property of the total momentum:
\begin{equation}
 k_{\mathrm{tot}} = \sum_{l=1}^{N_a} k_l
 = \frac{2\pi}{L} \sum_{l=1}^{N_a} I_l + \frac{1}{L} \sum_{l=1}^{N_a} \sum_{l'=1}^{N_a} \Theta(k_l, k_{l'}) \equiv
 \begin{cases}
  0 & N_a = \text{even}, \\
  \pi & N_a = \text{odd},
 \end{cases}
\end{equation}
where $k \equiv k'$ means that $k$ is equivalent to $k'$ modulo a reciprocal lattice vector.

\section{Relation between $\mathbb{Z}_2$ rotation symmetry and crystal momentum}
In this section, we show the detailed proof of the following theorem on the one-dimensional antiferromagnetic chain with a transverse magnetic field and easy-plane anisotropy [see Eqs.~\eqref{eq:Hamiltonian_SM} and \eqref{eq:Hamiltonian_map_SM}].
\begin{theorem}
 \label{thm:Z2_momentum_SM}
 Let $L$ be an even number of sites.
 Let $\hat{Z} = \bigotimes_{j=1}^{L} \hat{Z}_j$ be $\pi$-rotation symmetry of spins with respect to the $z$ axis such that $(\hat{Z}_j)^2 = 1$ (the precise definition of $\hat{Z}$ will be shown later).
 If the ground state is a simultaneous eigenstate of $\hat{Z}$ with a nontrivial eigenvalue $-1$, the state has a crystal momentum $\pi$; otherwise, the crystal momentum of the ground state is zero.
\end{theorem}

\subsection{The original model}
\label{sec:Z2_original}
First we consider the original model,
\begin{equation}
 \hat{\mathcal{H}} = J \sum_{j=1}^{L} \hat{\bm{S}}_j \cdot \hat{\bm{S}}_{j+1} + K \sum_{j=1}^{L} (\hat{S}_j^y)^2 - H \sum_{j=1}^{L} \hat{S}_j^z.
 \label{eq:Hamiltonian_SM}
\end{equation}
In the following, we assume the parameter regime with $J, K > 0$, $H \neq 0$, and $S \neq 1/2$, in which the total spin $\hat{S}_{\mathrm{tot}}^z = \sum_{j=1}^{L} \hat{S}_j^z$ is not conserved but the $\pi$-rotation symmetry $\hat{Z} := \bigotimes_{j=1}^{L} \ee^{\ii \pi (S - \hat{S}_j^z)}$ exists.
As a prerequisite for proving Theorem~\ref{thm:Z2_momentum_SM}, let us introduce the following Definition and two Propositions.
\begin{definition}[signed basis]
 Let $\ket{\bm{m}} := \ket{m_1 m_2 \dots m_L} \ (m_j = -S, -S+1, \dots, S)$ be a conventional basis of the Hilbert space that is an eigenstate of $\hat{S}_j^z$:
 \begin{equation}
  \hat{S}_j^z \ket{\bm{m}} = m_j \ket{\bm{m}}, \quad
  \hat{S}_j^{\pm} \ket{\bm{m}} = \sqrt{S (S + 1) - m_j (m_j \pm 1)} \ket{m_1, \dots, m_{j-1}, m_j \pm 1, m_{j+1}, \dots, m_L}
  =: c_{\pm}(m_j) \ket{\bar{\bm{m}}^{(j, \pm 1)}},
 \end{equation}
 for all $j \in \{1, 2, \dots, L\}$.
 Then we define a new set of basis called a \textit{signed basis}, which is given by $\ket{\tilde{\bm{m}}} := (-1)^{\delta(\bm{m})} \ket{\bm{m}}$ with
 \begin{equation}
  \delta(\bm{m}) := \sum_{j=1}^{L} j (S - m_j).
  \label{eq:sign}
 \end{equation}
\end{definition}

\begin{proposition}[off-diagonal matrix element]
 \label{prop:off-diagonal_original}
 All the off-diagonal elements of $\hat{\mathcal{H}}$ in the signed basis are nonpositive\textup{:}
 \begin{equation}
  \braket{\tilde{\bm{m}} | \hat{\mathcal{H}} | \tilde{\bm{m}}'} \leq 0,
  \ \forall \bm{m}, \bm{m}' \ \text{s.t.} \ \bm{m} \neq \bm{m}'.
  \label{eq:H_nonpositive_original}
 \end{equation}
\end{proposition}
\begin{proof}
 The Hamiltonian Eq.~\eqref{eq:Hamiltonian_SM} is rewritten by using the ladder operators $\hat{S}_j^{\pm} = \hat{S}_j^x \pm \ii \hat{S}_j^y$:
 \begin{equation}
  \hat{\mathcal{H}}
  = \sum_{j=1}^{L} \left\{ J \left[ \frac{1}{2} (\hat{S}_j^+ \hat{S}_{j+1}^- + \hat{S}_j^- \hat{S}_{j+1}^+) + \hat{S}_j^z \hat{S}_{j+1}^z \right] - \frac{K}{4} \left[ (\hat{S}_j^+)^2 + (\hat{S}_j^-)^2 - \hat{S}_j^+ \hat{S}_j^- - \hat{S}_j^- \hat{S}_j^+ \right] - H \hat{S}_j^z \right\}
  =: \sum_{j=1}^{L} \hat{h}_j.
 \end{equation}
 The local operator $\hat{h}_j$ acts on the conventional basis as
 \begin{align}
  \hat{h}_j \ket{\bm{m}}
  &= \frac{J}{2} \left[ c_+(m_j) c_-(m_{j+1}) \ket{\bar{\bm{m}}^{(j, +1; j+1, -1)}} + c_-(m_j) c_+(m_{j+1}) \ket{\bar{\bm{m}}^{(j, -1; j+1, +1)}} \right] \notag \\
  &\quad - \frac{K}{4} \left[ c_+(m_j + 1) c_+(m_j) \ket{\bar{\bm{m}}^{(j, +2)}} + c_-(m_j - 1) c_-(m_j) \ket{\bar{\bm{m}}^{(j, -2)}} \right] \notag \\
  &\quad + \left\{ \frac{K}{4} \left[ c_+(m_j - 1) c_-(m_j) + c_-(m_j + 1) c_+(m_j) \right] + (J m_{j+1} - H) m_j \right\} \ket{\bm{m}}.
 \end{align}
 Here, using Eq.~\eqref{eq:sign}, we can obtain the following relations:
 \begin{align}
  (-1)^{\delta(\bar{\bm{m}}^{(j, \pm 1; j+1, \mp 1)})} &= (-1)^{\delta(\bm{m}) \pm 1} = - (-1)^{\delta(\bm{m})}, \\
  (-1)^{\delta(\bar{\bm{m}}^{(j, \pm 2)})} &= (-1)^{\delta(\bm{m}) \mp 2j} = (-1)^{\delta(\bm{m})}.
 \end{align}
 Therefore, the local Hamiltonian acts on the signed basis as
 \begin{align}
  \hat{h}_j \ket{\tilde{\bm{m}}}
  &= \textcolor{red}{-} \frac{J}{2} \left[ c_+(m_j) c_-(m_{j+1}) \ket{\tilde{\bar{\bm{m}}^{(j, +1; j+1, -1)}}} + c_-(m_j) c_+(m_{j+1}) \ket{\tilde{\bar{\bm{m}}^{(j, -1; j+1, +1)}}} \right] \notag \\
  &\quad - \frac{K}{4} \left[ c_+(m_j + 1) c_+(m_j) \ket{\tilde{\bar{\bm{m}}^{(j, +2)}}} + c_-(m_j - 1) c_-(m_j) \ket{\tilde{\bar{\bm{m}}^{(j, -2)}}} \right] \notag \\
  &\quad + \left\{ \frac{K}{4} \left[ c_+(m_j - 1) c_-(m_j) + c_-(m_j + 1) c_+(m_j) \right] + (J m_{j+1} - H) m_j \right\} \ket{\tilde{\bm{m}}}.
  \label{eq:local_op_signed}
 \end{align}
 Since $J$, $K$, and $c_{\pm}(m)$ are all nonnegative, the off-diagonal elements of $\hat{h}_j$, and those of $\hat{\mathcal{H}}$, in the signed basis are nonpositive.
\end{proof}

\begin{proposition}[irreducibility]
 \label{prop:irreducibility_original}
 Since the Hamiltonian $\hat{\mathcal{H}}$ has the $\mathbb{Z}_2$ symmetry and thus commutes with $\hat{Z}$, it can be block-diagonalized by eigensectors of $\hat{Z} = \pm 1$\textup{:} $\hat{\mathcal{H}} = \hat{\mathcal{H}}^{(+)} \oplus \hat{\mathcal{H}}^{(-)}$.
 Let $V^{(\pm)}$ be a vector space with its eigenvalue of $\hat{Z}$ being $\pm 1$: $V^{(\pm)} := \left\{\bm{m} \relmiddle| \hat{Z} \ket{\tilde{\bm{m}}} = \pm \ket{\tilde{\bm{m}}}\right\}$.
 Then the following statement holds\textup{:}
 \begin{equation}
  \forall \bm{m}, \bm{m}' \in V^{(\pm)} \ \text{s.t.} \ \bm{m} \neq \bm{m}',
  \exists l \in \mathbb{N} \ \text{s.t.} \
  \braket{\tilde{\bm{m}} | (- \hat{\mathcal{H}}^{(\pm)})^l | \tilde{\bm{m}}'} \neq 0.
  \label{eq:H_irreducibility_original}
 \end{equation}
\end{proposition}

\begin{figure}[tbp]
 \includegraphics[width=.9\linewidth, pagebox=artbox]{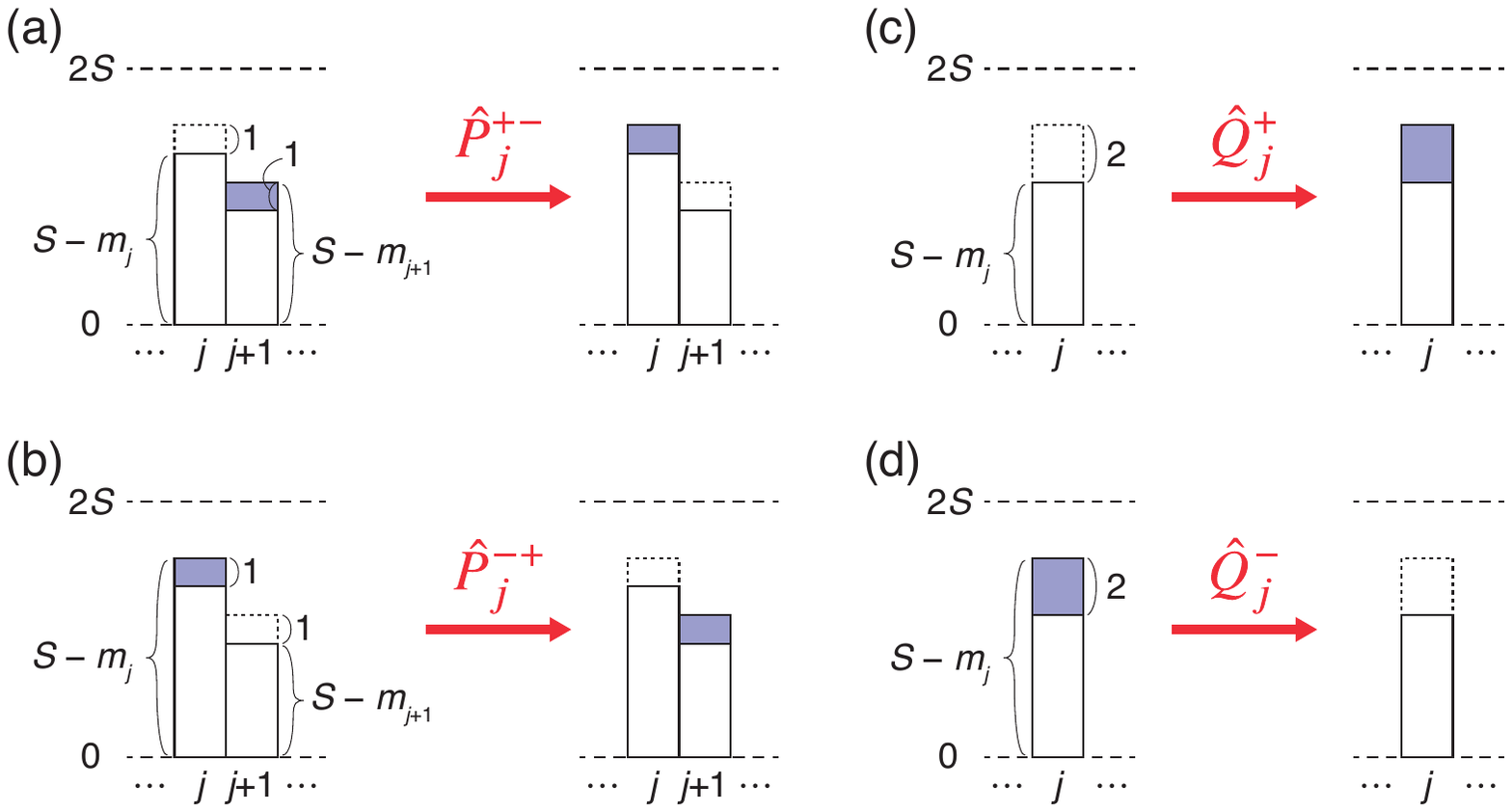}
 \caption{Schematic illustrations of the fundamental operators (a) $\hat{P}_j^{+-}$, (b) $\hat{P}_j^{-+}$, (c) $\hat{Q}_j^{+}$, and (d) $\hat{Q}_j^{-}$. The length of the $j$th rectangle shows $S - m_j = 0, 1, \dots, 2S$. The coefficients in Eqs.~\eqref{eq:fundamental_op_P} and \eqref{eq:fundamental_op_Q} are not reflected on these figures.}
 \label{fig:fundamental_op}
\end{figure}

For the proof of Proposition~\ref{prop:irreducibility_original}, we decompose the off-diagonal part, with respect to the signed basis, of the local operator $\hat{h}_j$ [Eq.~\eqref{eq:local_op_signed}] as
\begin{equation}
 - \hat{h}_j = \hat{P}_j^{+-} + \hat{P}_j^{-+} + \hat{Q}_j^{+} + \hat{Q}_j^{-} + \text{(diagonal part)},
\end{equation}
where the four \textit{fundamental operators} are defined by
\begin{align}
 \hat{P}_j^{\mp\pm}
 &:= \frac{J}{2} c_{\pm}(m_j) c_{\mp}(m_{j+1}) \ket{\tilde{\bar{\bm{m}}^{(j, \pm 1; j+1, \mp 1)}}} \bra{\tilde{\bm{m}}},
 \label{eq:fundamental_op_P} \\
 \hat{Q}_j^{\mp}
 &:= \frac{K}{4} c_{\pm}(m_j \pm 1) c_{\pm}(m_j) \ket{\tilde{\bar{\bm{m}}^{(j, \pm 2)}}} \bra{\tilde{\bm{m}}}.
 \label{eq:fundamental_op_Q}
\end{align}
These fundamental operators are schematically shown in Fig.~\ref{fig:fundamental_op}.
$\hat{P}_j^{+-}$ and $\hat{P}_j^{-+}$ do not alter $n_{\mathrm{tot}} := \sum_{j=1}^{L} (S - m_j)$ [Figs.~\ref{fig:fundamental_op}(a) and \ref{fig:fundamental_op}(b)], whereas $\hat{Q}_j^{+}$ and $\hat{Q}_j^{-}$ change $n_{\mathrm{tot}}$ to $n_{\mathrm{tot}} \pm 2$ [Figs.~\ref{fig:fundamental_op}(c) and \ref{fig:fundamental_op}(d)], respectively.
Using the operators, Eq.~\eqref{eq:H_irreducibility_original} is restated as follows: \textit{for any choices of $\bm{m} \neq \bm{m}' \in V^{(\pm)}$, there exists a combination of the fundamental operators that has a nonzero matrix element between $\bra{\tilde{\bm{m}}}$ and $\ket{\tilde{\bm{m}}'}$}~\footnote{There is no need to consider the cancellation of different matrix elements since Proposition~\ref{prop:off-diagonal_original} ensures the absence of off-diagonal matrix elements with opposite signs.}.
Therefore, Proposition~\ref{prop:irreducibility_original} is equivalent to the following statement.

\begin{propositionp}{\ref*{prop:irreducibility_original}$'$}
 \label{prop:irreducibility_original_2}
 We use the alternative representation $\bm{n} = [n_1, n_2, \dots, n_L]^{\mathrm{T}}$ instead of $\bm{m}$, where $n_j := S - m_j = 0, 1, \dots, 2S$.
 Let $V^{(+)} \ (V^{(-)})$ be a set of such vectors whose total number $n_{\mathrm{tot}} := \sum_{j=1}^{L} n_j$ is even (odd)\textup{:} $V^{(\pm)} := \{\bm{n} \mid n_j = 0, \dots, 2S; (-1)^{n_{\mathrm{tot}}} = \pm 1\}$.
 Let $\hat{P}_j^{\pm\mp}$ and $\hat{Q}_j^{\pm}$ $(j = 1, 2, \dots, L)$ be fundamental operators acting on $\bm{n}$ as follows\textup{:}
 \begin{alignat}{2}
  \hat{P}_j^{+-}\colon & \bm{n} \mapsto [n_1, \dots, n_{j-1}, n_j + 1, n_{j+1} - 1, n_{j+2}, \dots, n_L]^{\mathrm{T}}
  & \quad & (n_j < 2S, \, n_{j+1} > 0),
  \label{eq:P+-_condition} \\
  \hat{P}_j^{-+}\colon & \bm{n} \mapsto [n_1, \dots, n_{j-1}, n_j - 1, n_{j+1} + 1, n_{j+2}, \dots, n_L]^{\mathrm{T}}
  & & (n_j > 0, \, n_{j+1} < 2S),
  \label{eq:P-+_condition} \\
  \hat{Q}_j^{+}\colon & \bm{n} \mapsto [n_1, \dots, n_{j-1}, n_j + 2, n_{j+1}, \dots, n_L]^{\mathrm{T}}
  & & (n_j < 2S - 1),
  \label{eq:Q+_condition} \\
  \hat{Q}_j^{-}\colon & \bm{n} \mapsto [n_1, \dots, n_{j-1}, n_j - 2, n_{j+1}, \dots, n_L]^{\mathrm{T}}
  & & (n_j > 1),
  \label{eq:Q-_condition}
 \end{alignat}
 where we consider the periodic boundary condition $n_{L+1} = n_1$.
 Note that every operator is defined only when the condition denoted in the parentheses is satisfied.
 When the domain of the map is $V^{(+)} \ (V^{(-)})$, the codomain is also $V^{(+)} \ (V^{(-)})$.
 Then the following statement holds\textup{:} for any $\bm{n} \neq \bm{n}' \in V^{(\pm)}$, there exists a combination of the fundamental operators that maps $\bm{n}$ to $\bm{n}'$.
\end{propositionp}

\begin{proof}
 We prove Proposition~\ref{prop:irreducibility_original_2} instead of Proposition~\ref{prop:irreducibility_original}.
 For this purpose, we first show the following Lemma.
 \begin{lemma}
  \label{lem:operate_R}
  For any choices of two indices $j_1 \neq j_2$ such that $n_{j_1} > 0$ and $n_{j_2} < 2S$, there exists an operator $\hat{R}_{j_1, j_2}$ constructed by products of $\hat{P}^{\pm\mp}_j$ and $\hat{Q}^{\pm}_j$, which acts on $\bm{n}$ as follows\textup{:}
  \begin{equation}
   \hat{R}_{j_1, j_2}\colon
   \begin{cases}
    [n_1, \dots, n_{j_1}, \dots, n_{j_2}, \dots, n_L]^{\mathrm{T}} \mapsto [n_1, \dots, n_{j_1} - 1, \dots, n_{j_2} + 1, \dots, n_L]^{\mathrm{T}} & (j_1 < j_2), \\
    [n_1, \dots, n_{j_2}, \dots, n_{j_1}, \dots, n_L]^{\mathrm{T}} \mapsto [n_1, \dots, n_{j_2} + 1, \dots, n_{j_1} - 1, \dots, n_L]^{\mathrm{T}} & (j_1 > j_2).
   \end{cases}
  \end{equation}
 \end{lemma}
 \noindent
 We can confirm this Lemma by explicitly constructing the operator:
 \begin{equation}
  \hat{R}_{j_1, j_2} =
  \begin{cases}
   \displaystyle
   \prod_{j \in \Lambda_0} \hat{Q}^{-}_j \cdot \hat{P}^{-+}_{j_1} \hat{P}^{-+}_{j_1 + 1} \dotsm \hat{P}^{-+}_{j_2 - 2} \hat{P}^{-+}_{j_2 - 1} \cdot \prod_{j \in \Lambda_0} \hat{Q}^{+}_j & (j_1 < j_2), \\
   \displaystyle
   \prod_{j \in \Lambda_0} \hat{Q}^{-}_j \cdot \hat{P}^{+-}_{j_1 - 1} \hat{P}^{+-}_{j_1 - 2} \dotsm \hat{P}^{+-}_{j_2 + 1} \hat{P}^{+-}_{j_2} \cdot \prod_{j \in \Lambda_0} \hat{Q}^{+}_j & (j_1 > j_2),
  \end{cases}
  \label{eq:operate_R}
 \end{equation}
 where $\Lambda_0 := \{j \mid n_j = 0, \, \min\{j_1, j_2\} < j < \max\{j_1, j_2\}\}$.
 In Fig.~\ref{fig:operate_R_ex}, for example, we show the processes in $\hat{R}_{2, 7}$ acting on $\bm{n} = [0, 1, 3, 0, 1, 0, 2, 1]^{\mathrm{T}}$ for $S = 3/2$ systems to intuitively understand Eq.~\eqref{eq:operate_R}.
 The image of the map is $[0, 0, 3, 0, 1, 0, 3, 1]^{\mathrm{T}}$ [Fig.~\ref{fig:operate_R_ex}(e)], which some would regard as obtainable by using $\hat{P}^{-+}_2 \dotsm \hat{P}^{-+}_6$.
 However, the initial vector has zero components at $j \in \Lambda_0 = \{4, 6\}$ between $j_1 (= 2)$ and $j_2 (= 7)$, which prevents directly multiplying $\hat{P}^{-+}_2 \dotsm \hat{P}^{-+}_6$ [see the condition in Eq.~\eqref{eq:P+-_condition}].
 Therefore, we first raise the zero elements by two using $\hat{Q}^{+}_4 \hat{Q}^{+}_6$ [Figs.~\ref{fig:operate_R_ex}(a)--\ref{fig:operate_R_ex}(b)], and then operate $\hat{P}^{-+}_2 \dotsm \hat{P}^{-+}_6$ [Figs.~\ref{fig:operate_R_ex}(b)--\ref{fig:operate_R_ex}(d)]; finally, we offset the increases by acting $\hat{Q}^{-}_4 \hat{Q}^{-}_6$ [Figs.~\ref{fig:operate_R_ex}(d)--\ref{fig:operate_R_ex}(e)].
 The construction of the operator $\hat{R}_{j_1, j_2}$ is also represented by Algorithm~\ref{alg:operate_R}.

 Using Lemma~\ref{lem:operate_R}, let us prove Proposition~\ref{prop:irreducibility_original_2}.
 Let $\Delta\bm{n}$ and $\Delta n_{\mathrm{tot}}$ be $\bm{n}' - \bm{n}$ and $n_{\mathrm{tot}}' - n_{\mathrm{tot}}$, respectively.
 Then we construct an algorithm to transform $\bm{n}$ to $\bm{n}'$ by combinations of the fundamental operators, which is composed of two steps.
 \begin{enumerate}[label=\underline{Step (\roman*)}, leftmargin=15mm]
  \item \textit{If $\Delta n_{\mathrm{tot}} \neq 0$, transform $\bm{n}$ until $\Delta n_{\mathrm{tot}} = 0$ is satisfied.} \\
  First, note that $\Delta n_{\mathrm{tot}}$ must be even since $\bm{n}$ and $\bm{n}'$ belong to the same set $V^{(+)}$ or $V^{(-)}$.
  We repeat one of the two operations (i-1) and (i-2) depending on the sign of $\Delta n_{\mathrm{tot}}$.
  \begin{enumerate}[label=(\roman{enumi}-\arabic*)]
   \item When $\Delta n_{\mathrm{tot}} < 0$, we need to decrease $n_{\mathrm{tot}}$.
   Thus we search for indices $j_1 < j_2 < \dotsb$ of the maximum components in $\bm{n}$: $n_{j_1} = n_{j_2} = \dots = \max_{1 \leq j \leq L} n_j$.
   If $n_{j_1} \geq 2$, we operate $\hat{Q}^{-}_{j_1}$ and reduce the component by two.
   Otherwise, namely $n_{j_1} = 1$, there exists at least one different index $j_2$ such that $n_{j_2} = 1$ because of the even $\Delta n_{\mathrm{tot}}$; in this case, we can transform $(n_{j_1}, n_{j_2}) = (1, 1)$ to $(2, 0)$ by acting $\hat{R}_{j_2, j_1}$.
   \item When $\Delta n_{\mathrm{tot}} > 0$, we need to increase $n_{\mathrm{tot}}$.
   Thus we search for indices $j_1 < j_2 < \dotsb$ of the minimum components in $\bm{n}$: $n_{j_1} = n_{j_2} = \dots = \min_{1 \leq j \leq L} n_j$.
   If $n_{j_1} \leq 2S - 2$, we operate $\hat{Q}^{+}_{j_1}$ and raise the component by two.
   Otherwise, namely $n_{j_1} = 2S - 1$, there exists at least one different index $j_2$ such that $n_{j_2} = 2S - 1$; in this case, we can transform $(n_{j_1}, n_{j_2}) = (2S - 1, 2S - 1)$ to $(2S - 2, 2S)$ by acting $\hat{R}_{j_1, j_2}$.
  \end{enumerate}
  \item \textit{If $\Delta\bm{n} \neq \bm{0}$, transform $\bm{n}$ until $\Delta\bm{n} = \bm{0}$ is satisfied.} \\
  Let $j_{\mathrm{max}}$ ($j_{\mathrm{min}}$) be the smallest index in $\argmax_{1 \leq j \leq L} \Delta n_j$ ($\argmin_{1 \leq j \leq L} \Delta n_j$).
  Since $\Delta n_{\mathrm{tot}} = 0$ due to the previous step, $\max_{1 \leq j \leq L} \Delta n_j$ and $\min_{1 \leq j \leq L} \Delta n_j$ have opposite signs as long as $\Delta\bm{n} \neq \bm{0}$.
  Therefore, we repeat acting $\hat{R}_{j_{\mathrm{min}}, j_{\mathrm{max}}}$ to $\bm{n}$ and finally get $\bm{n} = \bm{n}'$.
 \end{enumerate}
 The above procedures are represented in Algorithm~\ref{alg:main}.
 In summary, we can transform $\bm{n}$ to $\bm{n}'$ by multiple actions of the fundamental operators, whose construction method is shown in Algorithms~\ref{alg:operate_R} and \ref{alg:main}~\footnote{The algorithms are indeed implemented in Julia: \href{https://github.com/shuntarosumita/antiferro_spin_parity/tree/main/irreducibility}{https://github.com/shuntarosumita/antiferro\_spin\_parity/tree/main/irreducibility}}.
 This completes the proof.
\end{proof}

\begin{figure}[tbp]
 \includegraphics[width=.8\linewidth, pagebox=artbox]{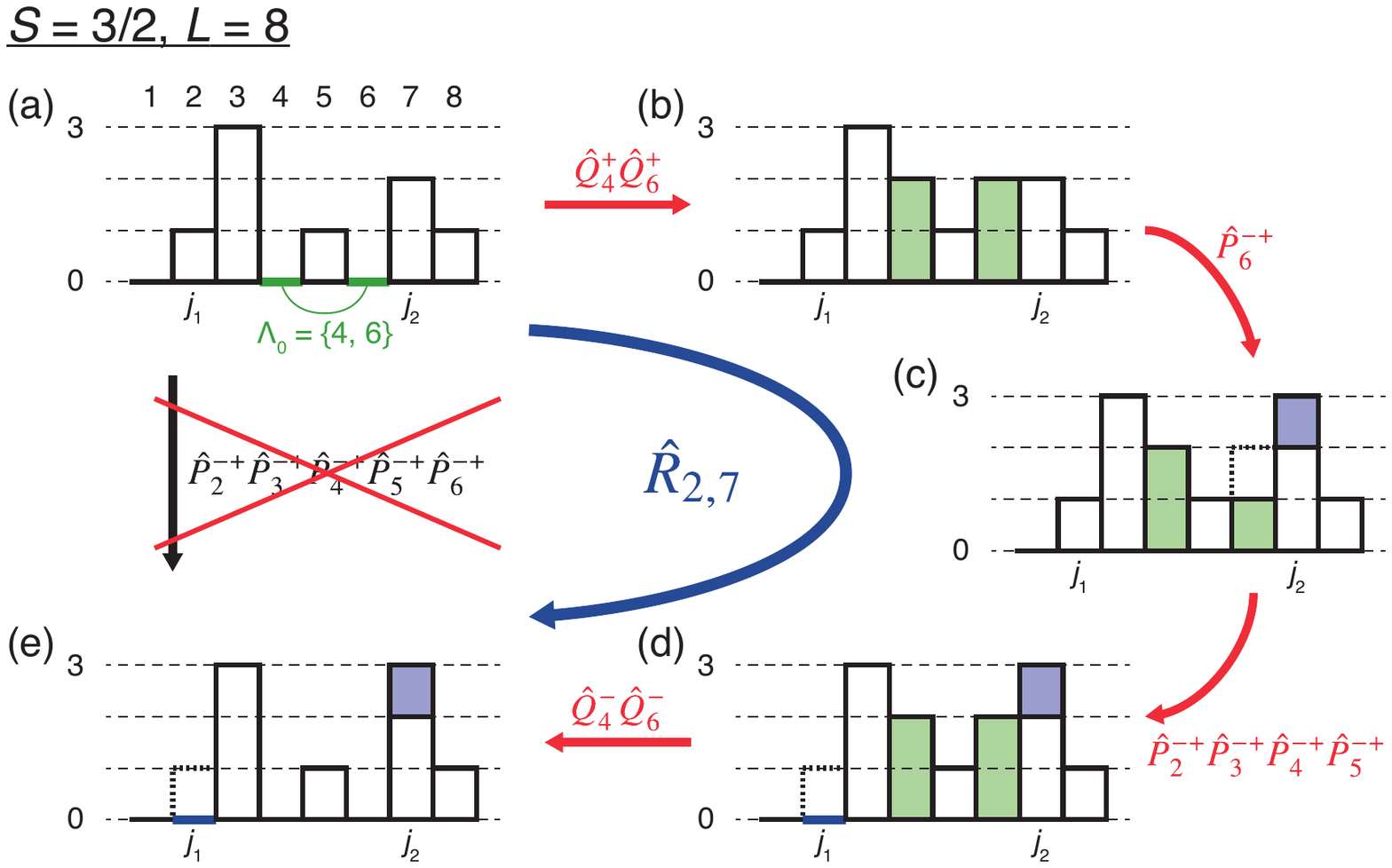}
 \caption{Schematic illustrations of the operator $\hat{R}_{j_1, j_2}$ for $S = 3/2$, $j_1 = 2$, and $j_2 = 7$. We take an input $\bm{n} = [0, 1, 3, 0, 1, 0, 2, 1]^{\mathrm{T}}$ as an example, and obtain the final result $\hat{R}_{2, 7}(\bm{n}) = [0, 0, 3, 0, 1, 0, 3, 1]^{\mathrm{T}}$.}
 \label{fig:operate_R_ex}
\end{figure}

\begin{table}[tbp]
 \begin{minipage}[t]{.45\linewidth}
  \begin{algorithm}[H]
   \caption{Implementation of Eq.~\eqref{eq:operate_R}.}
   \label{alg:operate_R}
   \begin{algorithmic}[1]
    \Require{$n_{j_1} > 0, \, n_{j_2} < 2S, \, j_1 \neq j_2$}
    \Function{operate\_R}{$\bm{n}^{\mathrm{in}}, j_1, j_2$}
     \State{$\Lambda_0 \gets \{j \mid n^{\mathrm{in}}_j = 0, \, \min\{j_1, j_2\} < j < \max\{j_1, j_2\}\}$}
     \State{$\bm{n}^{\mathrm{out}} \gets \bm{n}^{\mathrm{in}}$}
     \State
     \ForAll{$j \ \text{in} \ \Lambda_0$}
      \State{$n^{\mathrm{out}}_j \gets n^{\mathrm{out}}_j + 2$}
      \Comment{$\hat{Q}_j^{+}$}
     \EndFor
     \State
     \If{$j_1 < j_2$}
      \For{$j = j_2 - 1, j_2 - 2, \dots, j_1 + 1, j_1$}
       \State{$n^{\mathrm{out}}_j \gets n^{\mathrm{out}}_j - 1$; $n^{\mathrm{out}}_{j+1} \gets n^{\mathrm{out}}_{j+1} + 1$}
       \Comment{$\hat{P}_j^{-+}$}
      \EndFor
     \Else
      \For{$j = j_2, j_2 + 1, \dots, j_1 - 2, j_1 - 1$}
       \State{$n^{\mathrm{out}}_j \gets n^{\mathrm{out}}_j + 1$; $n^{\mathrm{out}}_{j+1} \gets n^{\mathrm{out}}_{j+1} - 1$}
       \Comment{$\hat{P}_j^{+-}$}
      \EndFor
     \EndIf
     \State
     \ForAll{$j \ \text{in} \ \Lambda_0$}
      \State{$n^{\mathrm{out}}_j \gets n^{\mathrm{out}}_j - 2$}
      \Comment{$\hat{Q}_j^{-}$}
     \EndFor
     \State
     \State{\Return{$\bm{n}^{\mathrm{out}}$}}
    \EndFunction
   \end{algorithmic}
  \end{algorithm}
  \begin{remark}
   ``$a \gets b$'' in the algorithms means ``substitute $b$ into $a$.'' Also, ``$\argmax$'' and ``$\argmin$'' functions are defined as follows:
   \begin{align*}
    \argmax_{j} n_j &:= \{j \mid \forall j': n_{j'} \leq n_{j}\}, \\
    \argmin_{j} n_j &:= \{j \mid \forall j': n_{j'} \geq n_{j}\}.
   \end{align*}
  \end{remark}
 \end{minipage}
 \hspace{10mm}
 \begin{minipage}[t]{.45\linewidth}
  \begin{algorithm}[H]
   \caption{A main algorithm for transforming $\bm{n}$ to $\bm{n}'$.}
   \label{alg:main}
   \begin{algorithmic}[1]
    \Require{$S, L, \bm{n}, \bm{n}'$}
    \State{$\Delta \bm{n} \gets \bm{n}' - \bm{n}$; $\Delta n_{\mathrm{tot}} \gets \sum_{j=1}^{L} \Delta n_j$}
    \While {$\Delta n_{\mathrm{tot}} \neq 0$}
     \If{$\Delta n_{\mathrm{tot}} < 0$}
      \State{$\Lambda_{\mathrm{max}} \gets \argmax_{1 \leq j \leq L} n_j$}
      \State{$j_1 \gets \min(\Lambda_{\mathrm{max}})$}
      \If{$n_{j_1} \geq 2$}
       \State{$n_{j_1} \gets n_{j_1} - 2$}
       \Comment{$\hat{Q}_{j_1}^{-}$}
      \Else
       \State{$j_2 \gets \min(\Lambda_{\mathrm{max}} \setminus \{j_1\})$}
       \State{$\bm{n} \gets \text{\Call{operate\_R}{$\bm{n}, j_2, j_1$}}$}
       \Comment{$\hat{R}_{j_2, j_1}$}
      \EndIf
     \Else
      \State{$\Lambda_{\mathrm{min}} \gets \argmin_{1 \leq j \leq L} n_j$}
      \State{$j_1 \gets \min(\Lambda_{\mathrm{min}})$}
      \If{$n_{j_1} \leq 2S - 2$}
       \State{$n_{j_1} \gets n_{j_1} + 2$}
       \Comment{$\hat{Q}_{j_1}^{+}$}
      \Else
       \State{$j_2 \gets \min(\Lambda_{\mathrm{min}} \setminus \{j_1\})$}
       \State{$\bm{n} \gets \text{\Call{operate\_R}{$\bm{n}, j_1, j_2$}}$}
       \Comment{$\hat{R}_{j_1, j_2}$}
      \EndIf
     \EndIf
     \State{$\Delta\bm{n} \gets \bm{n}' - \bm{n}$; $\Delta n_{\mathrm{tot}} \gets \sum_{j=1}^{L} \Delta n_j$}
    \EndWhile
    \State
    \While{$\Delta\bm{n} \neq \bm{0}$}
     \State{$j_{\mathrm{max}} \gets \min(\argmax_{1 \leq j \leq L} \Delta n_j)$}
     \State{$j_{\mathrm{min}} \gets \min(\argmin_{1 \leq j \leq L} \Delta n_j)$}
     \State{$\bm{n} \gets \text{\Call{operate\_R}{$\bm{n}, j_{\mathrm{min}}, j_{\mathrm{max}}$}}$}
     \Comment{$\hat{R}_{j_{\mathrm{min}}, j_{\mathrm{max}}}$}
     \State{$\Delta \bm{n} \gets \bm{n}' - \bm{n}$}
    \EndWhile
   \end{algorithmic}
  \end{algorithm}
 \end{minipage}
\end{table}

Using Propositions~\ref{prop:off-diagonal_original} and \ref{prop:irreducibility_original}, we prove the main theorem.
\begin{proof}[Proof of Theorem~\ref{thm:Z2_momentum_SM} for Eq.~\eqref{eq:Hamiltonian_SM}]
 According to Propositions~\ref{prop:off-diagonal_original} and \ref{prop:irreducibility_original} [Eqs.~\eqref{eq:H_nonpositive_original} and \eqref{eq:H_irreducibility_original}], we can apply the Perron--Frobenius theorem~\cite{Perron1907_SM, Frobenius1912_SM} to $\hat{\mathcal{H}}^{(\pm)}$.
 As a result, the ground state of each Hamiltonian block is nondegenerate, and is given by
 \begin{equation}
  \Ket{\Psi_{\mathrm{GS}}^{(\pm)}} = \sum_{\bm{m} \in V^{(\pm)}} a(\bm{m}) \ket{\tilde{\bm{m}}}, \quad
  a(\bm{m}) > 0.
 \end{equation}
 We now calculate the eigenvalue of the one-site translation operator $\hat{T}$.
 Noting that $(-1)^{L (S - m_L)} = 1$ since $L$ is even,
 \begin{align}
  (-1)^{\delta(\hat{T}(\bm{m}))}
  &= (-1)^{\sum_{j=1}^{L} j (S - m_{j-1})} \notag \\
  &= (-1)^{\sum_{j=1}^{L-1} (j+1) (S - m_j) + (S - m_L)} \notag \\
  &= (-1)^{\sum_{j=1}^{L} (j+1) (S - m_j)}
  = (-1)^{\sum_{j=1}^{L} (S - m_j)} (-1)^{\delta(\bm{m})}.
  \label{eq:sign_translation}
 \end{align}
 Therefore, $\hat{T}$ acts on the ground state as
 \begin{align}
  \hat{T} \Ket{\Psi_{\mathrm{GS}}^{(\pm)}}
  &= \sum_{\bm{m} \in V^{(\pm)}} a(\bm{m}) (-1)^{\delta(\bm{m})} \ket{\hat{T}(\bm{m})} \notag \\
  &= \sum_{\bm{m} \in V^{(\pm)}} a(\bm{m}) (-1)^{\sum_{j=1}^{L} (S - m_j)} (-1)^{\delta(\hat{T}(\bm{m}))} \ket{\hat{T}(\bm{m})} \notag \\
  &= \sum_{\bm{m} \in V^{(\pm)}} a(\hat{T}^{-1}(\bm{m})) \left[\prod_{j=1}^{L} (-1)^{S - m_j}\right] (-1)^{\delta(\bm{m})} \ket{\bm{m}}.
 \end{align}
 Because of the uniqueness of the ground state, $\Ket{\Psi_{\mathrm{GS}}^{(\pm)}}$ is an eigenstate of $\hat{T}$, that is, $|a(\hat{T}(\bm{m}))| = |a(\bm{m})|$.
 Since $a(\bm{m}) > 0$, this results in $a(\hat{T}(\bm{m})) = a(\bm{m})$.
 Thus we obtain
 \begin{equation}
  \hat{T} \Ket{\Psi_{\mathrm{GS}}^{(\pm)}}
  = \sum_{\bm{m} \in V^{(\pm)}} a(\bm{m}) \underbrace{\left[\prod_{j=1}^{L} (-1)^{S - m_j}\right] \ket{\tilde{\bm{m}}}}_{\hat{Z} \ket{\tilde{\bm{m}}} \, = \, \pm \ket{\tilde{\bm{m}}}}
  = \pm \Ket{\Psi_{\mathrm{GS}}^{(\pm)}}.
 \end{equation}
 This completes the proof.
\end{proof}

\subsection{The mapped spin-$1/2$ XYX model}
Let us move on to the spin-$1/2$ XYX model under a transverse field, which is obtained by the first-order perturbation with respect to the exchange and Zeeman interactions in half-odd-integer spin cases of Eq.~\eqref{eq:Hamiltonian_SM}; for details, see Sec.~\ref{sec:mapping_largeK}.
The mapped Hamiltonian is given by
\begin{equation}
 \hat{\mathcal{H}}_{\mathrm{map}} = \tilde{J} \sum_{j=1}^{L} (\hat{s}_j^x \hat{s}_{j+1}^x + \Delta \hat{s}_j^y \hat{s}_{j+1}^y + \hat{s}_j^z \hat{s}_{j+1}^z) - \tilde{H} \sum_{j=1}^{L} \hat{s}_j^z,
 \label{eq:Hamiltonian_map_SM}
\end{equation}
where $\hat{\bm{s}}_j$ is the mapped spin-$1/2$ operator on the $j$th site.
In the mapped model, the single-ion anisotropy is taken into the anisotropy of the exchange interaction $\Delta := 1 / (S + 1/2)^2$.
Note that $\tilde{J} := J / \Delta$ and $\tilde{H} := H / \sqrt{\Delta}$.
In the following, we focus on $\tilde{J} > 0$, $\tilde{H} \neq 0$, and $S \geq 3/2$ (i.e., $0 < \Delta < 1$)~%
\footnote{%
For $S = 1/2$ ($\Delta = 1$), the mapped model is equivalent to the original model, whose property of level crossings is described by gapless TLL, as we explain in the main text.
In this case, we do not need to consider the discrete $\pi$-rotation symmetry $\hat{Z}$ since the total spin $\sum_{j=1}^{L} \hat{S}_j^z$ is conserved.
}.
Let $\hat{Z} := \bigotimes_{j=1}^{L} \ee^{\ii \pi (1/2 - \hat{s}_j^z)} = \bigotimes_{j=1}^{L} (2\hat{s}_j^z)$ be the $\pi$-rotation symmetry.
To prove Theorem~\ref{thm:Z2_momentum_SM}, we present the following two Propositions, which correspond to Propositions~\ref{prop:off-diagonal_original} and \ref{prop:irreducibility_original} in the original model.

\begin{proposition}[off-diagonal matrix element]
 \label{prop:off-diagonal_map}
 All the off-diagonal elements of $\hat{\mathcal{H}}_{\mathrm{map}}$ in the signed basis are nonpositive\textup{:}
 \begin{equation}
  \braket{\tilde{\bm{m}} | \hat{\mathcal{H}}_{\mathrm{map}} | \tilde{\bm{m}}'} \leq 0
  \ \text{for all} \ \bm{m}, \bm{m}' \ \text{s.t.} \ \bm{m} \neq \bm{m}'.
  \label{eq:H_nonpositive_map}
 \end{equation}
\end{proposition}
\begin{proof}
 The Hamiltonian is rewritten by using ladder operators $\hat{s}_j^{\pm} = \hat{s}_j^x \pm \ii \hat{s}_j^y$:
 \begin{equation}
  \hat{\mathcal{H}}_{\mathrm{map}}
  = \sum_{j=1}^{L} \left\{ \tilde{J} \left[ \frac{1+\Delta}{4} (\hat{s}_j^+ \hat{s}_{j+1}^- + \hat{s}_j^- \hat{s}_{j+1}^+) + \frac{1-\Delta}{4} (\hat{s}_j^+ \hat{s}_{j+1}^+ + \hat{s}_j^- \hat{s}_{j+1}^-) + \hat{s}_j^z \hat{s}_{j+1}^z \right] - \tilde{H} \hat{s}_j^z \right\}
  =: \sum_{j=1}^{L} \hat{h}_{\mathrm{map}, j}.
 \end{equation}
 The local operator $\hat{h}_{\mathrm{map}, j}$ acts on the usual basis as
 \begin{align}
  \hat{h}_{\mathrm{map}, j} \ket{\bm{m}}
  &= \frac{\tilde{J}}{4} \left[ (1+\Delta) \left(\delta^{(j; j+1)}(-+) + \delta^{(j; j+1)}(+-)\right) + (1-\Delta) \left(\delta^{(j; j+1)}(--) + \delta^{(j; j+1)}(++)\right) \right] \ket{\bar{\bm{m}}^{(j; j+1)}} \notag \\
  &\quad + (\tilde{J} m_{j+1} - \tilde{H}) m_j \ket{\bm{m}},
 \end{align}
 where
 \begin{align}
  \delta^{(j; j+1)}(\pm \pm') &:=
  \begin{cases}
   1 & \text{if} \ m_j = \pm \frac{1}{2} \ \text{and} \ m_{j+1} = \pm' \frac{1}{2}, \\
   0 & \text{otherwise},
  \end{cases}
  \\
  \ket{\bar{\bm{m}}^{(j; j+1)}} &:= \ket{m_1 \dots m_{j-1} \bar{m}_j \bar{m}_{j+1} m_{j+2} \dots m_L}
  \quad (\bar{m}_j = - m_j).
 \end{align}
 Here, using Eq.~\eqref{eq:sign}, we can easily show that the sign $(-1)^{\delta(\bm{m})}$ changes when the $j$th and $(j+1)$th spins are flipped:
 \begin{equation}
  (-1)^{\delta(\bar{\bm{m}}^{(j; j+1)})}
  = (-1)^{\delta(\bm{m}) + 2j m_j + 2(j+1) m_{j+1}}
  = - (-1)^{\delta(\bm{m})},
 \end{equation}
 since both $2m_j$ and $2m_{j+1}$ are odd integers.
 Therefore, the local Hamiltonian acts on the signed basis $\ket{\tilde{\bm{m}}} = (-1)^{\delta(\bm{m})} \ket{\bm{m}}$ as
 \begin{align}
  \hat{h}_{\mathrm{map}, j} \ket{\tilde{\bm{m}}}
  &= \textcolor{red}{-} \frac{\tilde{J}}{4} \left[ (1+\Delta) \left(\delta^{(j; j+1)}(-+) + \delta^{(j; j+1)}(+-)\right) + (1-\Delta) \left(\delta^{(j; j+1)}(--) + \delta^{(j; j+1)}(++)\right) \right] \ket{\tilde{\bar{\bm{m}}^{(j; j+1)}}} \notag \\
  &\quad + (\tilde{J} m_{j+1} - \tilde{H}) m_j \ket{\tilde{\bm{m}}}.
 \end{align}
 Since $\tilde{J}$, $1+\Delta$, and $1-\Delta$ are all positive, the off-diagonal elements of $\hat{h}_{\mathrm{map}, j}$, and those of $\hat{\mathcal{H}}_{\mathrm{map}}$, in the signed basis are nonpositive.
\end{proof}

\begin{proposition}[irreducibility]
 \label{prop:irreducibility_map}
 The Hamiltonian $\hat{\mathcal{H}}_{\mathrm{map}}$ can be block-diagonalized by eigensectors of the $\mathbb{Z}_2$ symmetry $\hat{Z}$ with eigenvalues $\pm 1$\textup{:} $\hat{\mathcal{H}}_{\mathrm{map}} = \hat{\mathcal{H}}_{\mathrm{map}}^{(+)} \oplus \hat{\mathcal{H}}_{\mathrm{map}}^{(-)}$.
 Then the following statement holds\textup{:}
 \begin{equation}
  \text{for all} \ \bm{m}, \bm{m}' \ \text{s.t.} \ \bm{m} \neq \bm{m}',
  \text{there exists} \ l \in \mathbb{N} \ \text{s.t.} \
  \braket{\tilde{\bm{m}} | (- \hat{\mathcal{H}}_{\mathrm{map}}^{(\pm)})^l | \tilde{\bm{m}}'} \neq 0.
  \label{eq:H_irreducibility_map}
 \end{equation}
\end{proposition}
\begin{proof}
 We can prove Proposition~\ref{prop:irreducibility_map} in a way similar to the proof of Proposition~\ref{prop:irreducibility_original}.
\end{proof}

\begin{proof}[Proof of Theorem~\ref{thm:Z2_momentum_SM} for Eq.~\eqref{eq:Hamiltonian_map_SM}]
 Using Propositions~\ref{prop:off-diagonal_map} and \ref{prop:irreducibility_map} [Eqs.~\eqref{eq:H_nonpositive_map} and \eqref{eq:H_irreducibility_map}], we can straightforwardly prove the Theorem in the same way as the proof in the original model.
\end{proof}

\subsection{Extension to inversion symmetry}
\label{sec:momentum_inversion}
The original and mapped Hamiltonian operators commute with not only the $\pi$-rotation operator $\hat{Z}$ but also the inversion one $\hat{I}$:
\begin{equation}
 [\hat{\mathcal{H}}, \hat{I}] = [\hat{\mathcal{H}}_{\mathrm{map}}, \hat{I}] = 0,
\end{equation}
where $\hat{I}$ transforms an operator and a state on the $j$th site to those on the $(L-j+1)$th site, and $\hat{I}^2 = 1$.
Here we extend Theorem~\ref{thm:Z2_momentum_SM} for the $\pi$-rotation symmetry to the inversion symmetry.
\begin{corollary}
 If the ground state is a simultaneous eigenstate of $\hat{I}$ with a nontrivial eigenvalue $-1$, the state has a crystal momentum $\pi$; otherwise, the crystal momentum of the ground state is zero.
\end{corollary}
\begin{proof}
 We here consider the original model $\hat{\mathcal{H}}$; the proof for the mapped model $\hat{\mathcal{H}}_{\mathrm{map}}$ is straightforward.
 Let us revisit the Hamiltonian blocks $\hat{\mathcal{H}}^{(\pm)}$, which are the eigensectors of the $\pi$-rotation symmetry $\hat{Z} = \pm 1$, respectively.
 Since $\hat{I}$ commutes with $\hat{Z}$, the unique ground state $\Ket{\Psi_{\mathrm{GS}}^{(\pm)}}$ in each sector is a simultaneous eigenstate of $\hat{I}$ and $\hat{Z}$.
 Also, the irreducibility of $\hat{\mathcal{H}}^{(\pm)}$ (Proposition~\ref{prop:irreducibility_original}) ensures that the eigenvalue of $\hat{I}$ is uniquely determined in each block.
 The relation between the inversion eigenvalue and the momentum of the ground state is derived as follows.
 First, the sign of the basis is transformed under the inversion as
 \begin{align}
  (-1)^{\delta(\hat{I}(\bm{m}))}
  &= (-1)^{\sum_{j=1}^{L} j (S - m_{L-j+1})} \notag \\
  &= (-1)^{\sum_{j=1}^{L} (L-j+1) (S - m_j)} \notag \\
  &= (-1)^{\sum_{j=1}^{L} (L+1) (S - m_j)} (-1)^{\delta(\bm{m})} \notag \\
  &= (-1)^{\sum_{j=1}^{L} (S - m_j)} (-1)^{\delta(\bm{m})}
  = (-1)^{\delta(\hat{T}(\bm{m}))},
 \end{align}
 where we use Eq.~\eqref{eq:sign_translation}.
 Using $a(\hat{T}(\bm{m})) = a(\hat{I}(\bm{m})) = a(\bm{m})$, therefore, we obtain
 \begin{align}
  \hat{I} \Ket{\Psi_{\mathrm{GS}}^{(\pm)}}
  &= \sum_{\bm{m} \in V^{(\pm)}} a(\bm{m}) (-1)^{\delta(\bm{m})} \ket{\hat{I}(\bm{m})} \notag \\
  &= \sum_{\bm{m} \in V^{(\pm)}} a(\bm{m}) (-1)^{\delta(\hat{I}^{-1}(\bm{m}))} \ket{\bm{m}} \notag \\
  &= \sum_{\bm{m} \in V^{(\pm)}} a(\bm{m}) (-1)^{\delta(\hat{T}^{-1}(\bm{m}))} \ket{\bm{m}} \notag \\
  &= \sum_{\bm{m} \in V^{(\pm)}} a(\bm{m}) (-1)^{\delta(\bm{m})} \ket{\hat{T}(\bm{m})}
  = \hat{T} \Ket{\Psi_{\mathrm{GS}}^{(\pm)}}
  = \pm \Ket{\Psi_{\mathrm{GS}}^{(\pm)}}.
 \end{align}
 This completes the proof.
\end{proof}

\section{$\pi$-momentum state and entanglement}
\begin{figure}[tbp]
 \includegraphics[width=\linewidth, pagebox=artbox]{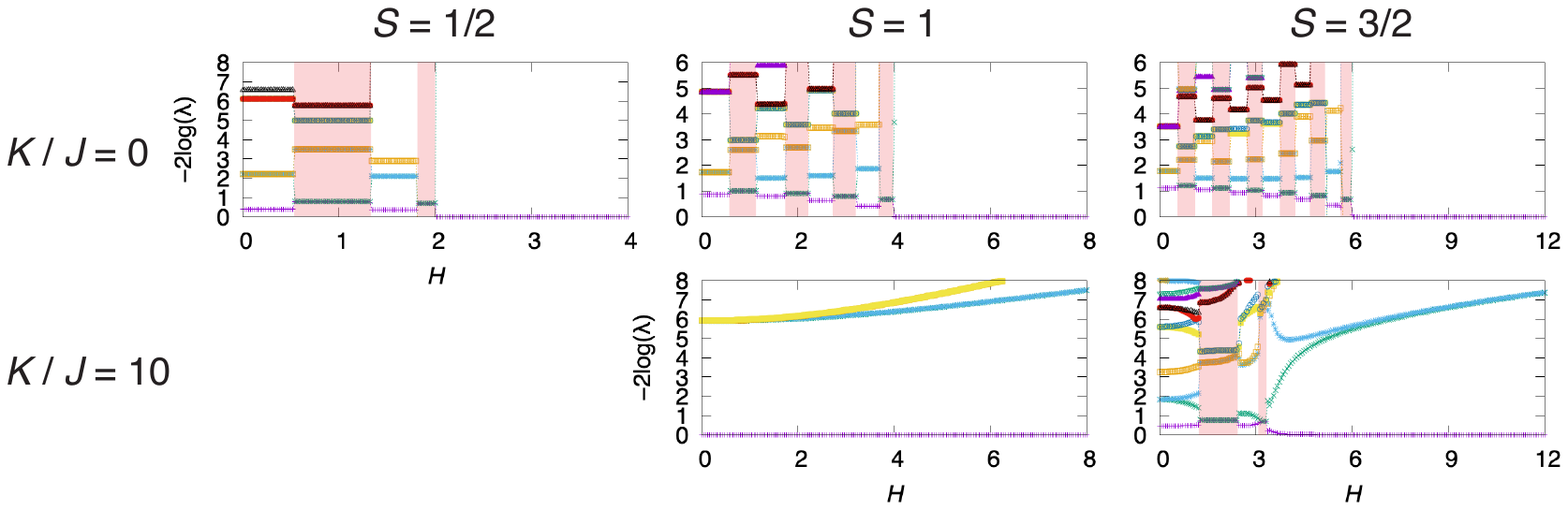}
 \caption{Ground-state entanglement spectra for $K / J = 0$ and $10$. $L$ is set to 8 for all panels. The spectrum has double degeneracy in the shaded regions where the crystal momentum of the ground state is $\pi$.}
 \label{fig:entanglement_spectrum}
\end{figure}

Here, we exhibit the relationship between crystal momentum and entanglement for the ground state.
First, we show entanglement spectra of the ground state for $K / J = 0$ and $10$ in Fig.~\ref{fig:entanglement_spectrum}, which are obtained by the density matrix renormalization group analyses using \texttt{ITensors.jl}~\cite{ITensor_SM, ITensor-r0.3_SM}.
In this figure, the shaded regions represent magnetic fields in which the ground state satisfies $\hat{Z} = -1$, namely $k = \pi$ (see Theorem~\ref{thm:Z2_momentum_SM}).
One can find that the entanglement spectrum is always doubly degenerate in the parameter regimes.

This double degeneracy can be understood as follows.
First, notice that the Hamiltonian Eq.~\eqref{eq:Hamiltonian_SM} commutes with $\hat{I}$ as well as $\hat{Z}$.
According to the symmetry arguments in Sec.~\ref{sec:momentum_inversion}, the eigenvalue of $\hat{I}$ is the same as that of $\hat{Z}$ for the ground state.
Following the discussions in Ref.~\cite{Pollmann2010_SM}, furthermore, one can reveal that the ground state with a nontrivial eigenvalue $\hat{I} = -1$ is characterized by the doubly degenerate entanglement spectrum.
The degeneracy may be relevant to a recent theoretical study in Ref.~\cite{Gioia2022_SM}, which suggested that the eigenstate with nonzero momentum must be long-range entangled.

\section{Mapping to spin-$1/2$ XYX model in large $K$ limit}
\label{sec:mapping_largeK}
In this section, we derive a spin-$1/2$ XYX model from Eq.~\eqref{eq:Hamiltonian_SM} for half-odd-integer $S$, through the first-order perturbation theory with respect to $J, H \ll K$.
Let us begin with the ``$K$ model'' $\hat{\mathcal{H}}_0 = K \sum_{j=1}^{L} (\hat{S}_j^y)^2$, which has twofold degenerate ground states $\ket{S_j^y = \pm\tfrac{1}{2}}$ on the $j$th site.
In the following, we consider two perturbation terms $\hat{H}_j' := -H \hat{S}_j^z$ and $\hat{H}_j'' := J \hat{\bm{S}}_j \cdot \hat{\bm{S}}_{j+1}$, restricting the Hilbert space to the ground state on every site.
First, we discuss the perturbation $\hat{H}_j'$ of a magnetic field.
Since
\begin{equation}
 \hat{S}_j^z \Ket{S_j^y = \pm\tfrac{1}{2}}
 = \frac{1}{2} \left[ \left(S + \frac{1}{2}\right) \Ket{S_j^y = \mp\tfrac{1}{2}} + \sqrt{\left(S - \frac{1}{2}\right)\left(S + \frac{3}{2}\right)} \Ket{S_j^y = \pm\tfrac{3}{2}} \right],
\end{equation}
the first-order perturbed Hamiltonian in the restricted Hilbert space is given by
\begin{equation}
 \begin{bmatrix}
  \Braket{S_j^y = +\tfrac{1}{2} | \hat{H}_j' | S_j^y = +\tfrac{1}{2}} & \Braket{S_j^y = +\tfrac{1}{2} | \hat{H}_j' | S_j^y = -\tfrac{1}{2}} \\[3mm]
  \Braket{S_j^y = -\tfrac{1}{2} | \hat{H}_j' | S_j^y = +\tfrac{1}{2}} & \Braket{S_j^y = -\tfrac{1}{2} | \hat{H}_j' | S_j^y = -\tfrac{1}{2}}
 \end{bmatrix}
 = - \frac{H}{2}
 \begin{bmatrix}
  0 & S + \frac{1}{2} \\
  S + \frac{1}{2} & 0
 \end{bmatrix}.
\end{equation}
Diagonalizing the matrix, we obtain the first-order energy $E_{\pm}^{(1)} = \mp (S + \frac{1}{2}) \frac{H}{2}$ and the corresponding wavefunction $\ket{\phi_{j, \pm}^{(0)}} = \left(\ket{S_j^y = +\tfrac{1}{2}} \pm \ket{S_j^y = -\tfrac{1}{2}}\right) / \sqrt{2}$ with the zeroth order in $H$.
The energy splitting is equivalent to a ``magnetic field'' $\tilde{H} := (S + \frac{1}{2}) H$ applied to a $1/2$ spin.

Next we consider the other perturbation $\hat{H}_j''$ with respect to the exchange interaction.
Using the bases $\ket{\phi_\pm^{(0)}}$, we obtain a matrix $[M_{j; \alpha \beta, \gamma \delta}] := \left[ \bra{\phi_{j, \alpha}^{(0)}} \bra{\phi_{j+1, \beta}^{(0)}} \hat{H}_j'' \ket{\phi_{j, \gamma}^{(0)}} \ket{\phi_{j+1, \delta}^{(0)}} \right]$ as follows:
\begin{equation}
 \begin{bmatrix}
  M_{j; ++, ++} & M_{j; ++, +-} & M_{j; ++, -+} & M_{j; ++, --} \\
  M_{j; +-, ++} & M_{j; +-, +-} & M_{j; +-, -+} & M_{j; +-, --} \\
  M_{j; -+, ++} & M_{j; -+, +-} & M_{j; -+, -+} & M_{j; -+, --} \\
  M_{j; --, ++} & M_{j; --, +-} & M_{j; --, -+} & M_{j; --, --}
 \end{bmatrix}
 = \frac{J}{4}
 \begin{bmatrix}
  \left(S + \frac{1}{2}\right)^2 & 0 & 0 & 1 - \left(S + \frac{1}{2}\right)^2 \\
  0 & - \left(S + \frac{1}{2}\right)^2 & 1 + \left(S + \frac{1}{2}\right)^2 & 0 \\
  0 & 1 + \left(S + \frac{1}{2}\right)^2 & - \left(S + \frac{1}{2}\right)^2 & 0 \\
  1 - \left(S + \frac{1}{2}\right)^2 & 0 & 0 & \left(S + \frac{1}{2}\right)^2
 \end{bmatrix}.
\end{equation}
These matrix elements are equivalent to those of a spin-$1/2$ XYX exchange interaction
\begin{equation}
 J \Bra{s_j^z = \tfrac{\alpha}{2}} \Bra{s_{j+1}^z = \tfrac{\beta}{2}}
 \left[ \left(S + \frac{1}{2}\right)^2 \hat{s}_j^x \hat{s}_{j+1}^x + \hat{s}_j^y \hat{s}_{j+1}^y + \left(S + \frac{1}{2}\right)^2 \hat{s}_j^z \hat{s}_{j+1}^z \right]
 \Ket{s_j^z = \tfrac{\gamma}{2}} \Ket{s_{j+1}^z = \tfrac{\delta}{2}}
 \quad (\alpha, \beta, \gamma, \delta = \pm 1),
\end{equation}
where $\hat{\bm{s}}_j$ is an effective spin-$1/2$ operator on the $j$th site.
Therefore, our model in Eq.~\eqref{eq:Hamiltonian_SM} can be mapped to the spin-$1/2$ XYX model in Eq.~\eqref{eq:Hamiltonian_map_SM} when $K \gg J, H$ is satisfied.

\begin{figure}[tbp]
 \includegraphics[width=.8\linewidth, pagebox=artbox]{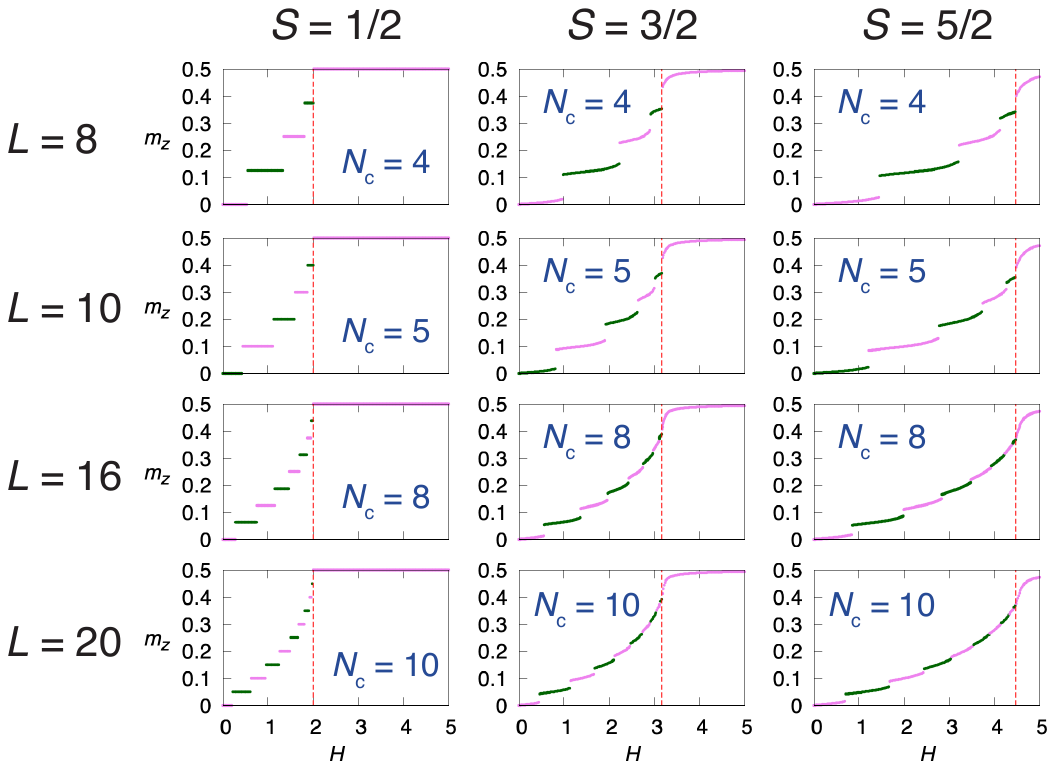}
 \caption{Magnetization $m_z$ vs magnetic field $H$ curves for $L = 8$, $10$, $16$, and $20$ obtained by exact diagonalization of the effective spin-$1/2$ model. The data for $L = 8$ are identical to those in the main text. The violet and green lines represent that the ground state has crystal momenta $0$ and $\pi$, respectively. The red dashed lines show the ``classical'' line $H = J \sqrt{2[(S + 1/2)^2 + 1]}$.}
 \label{fig:numerical_results_supple_xyx}
\end{figure}

We show the magnetization curves of the effective spin-$1/2$ model in Fig.~\ref{fig:numerical_results_supple_xyx}, where the first row ($L = 8$) is the same data as those in Fig.~1 of the main text.
All data exhibit level crossings with $N_{\mathrm{c}} = L/2$, consistent with the relations in Eq.~(2) in the main text.
It is also worth noting that in the large-$S$ limit, the Hamiltonian $\hat{\mathcal{H}}_{\mathrm{map}}$ reduces to an XX model in a transverse field, $\hat{\mathcal{H}}_{\mathrm{map}} = \tilde{J} \sum_{j=1}^{L} (\hat{s}_j^x \hat{s}_{j+1}^x + \hat{s}_j^z \hat{s}_{j+1}^z) - \tilde{H} \sum_{j=1}^{L} \hat{s}_j^z$, for which level crossings between $0$- and $\pi$-momentum states can also be confirmed numerically.
Since a \textit{classical} XX model in a transverse field just exhibits canted moments, we can conclude that the level crossings are a quantum effect.

\section{Exact solution of level crossing on a ``classical'' line}
In the main text and Sec.~\ref{sec:mapping_largeK}, the mapped spin-$1/2$ XYX (XXZ) model under a transverse magnetic field [Eq.~\eqref{eq:Hamiltonian_map_SM}] is discussed for half-odd-integer spins.
We here make contact with earlier studies of the effective model.
References~\cite{Kurmann1982_SM, Muller1985_SM, Dmitriev2002_SM} identified a ``classical'' line $\tilde{H} = \tilde{J} \sqrt{2(1 + \Delta)}$ on which the model becomes exactly solvable and a level crossing takes place between $0$- and $\pi$-momentum states.
We have checked that this line coincides with the ferromagnetic transition point (see red dashed lines in Fig.~\ref{fig:numerical_results_supple_xyx}).
In the following, to make our manuscript self-contained, we explicitly show that the effective model has twofold degenerate ground states on the ``classical'' line.
Note that the following discussions have some overlaps with Refs.~\cite{Kurmann1982_SM, Muller1985_SM, Dmitriev2002_SM}.

Let $\phi$ be $\arccos(\tilde{H}_{\mathrm{cl}} / 2\tilde{J}) = \arccos\left(\sqrt{(1 + \Delta) / 2}\right)$ (note that $0 < \Delta \leq 1$).
Then, we consider a unitary transformation $\hat{R}(\phi) := \bigotimes_{j=1}^{L} \exp[\ii (-1)^j \phi \hat{s}_j^y]$, which rotates a spin on the odd-$j$th (even-$j$th) site with an angle $\phi$ ($-\phi$) around the $y$ axis:
\begin{equation}
 \hat{R}(\phi)^{-1} \hat{\bm{s}}_j \hat{R}(\phi)
 =
 \underbrace{
 \begin{bmatrix}
  \cos\phi & 0 & -(-1)^j \sin\phi \\
  0 & 1 & 0 \\
  (-1)^j \sin\phi & 0 & \cos\phi
 \end{bmatrix}
 }_{U_j(\phi)}
 \hat{\bm{s}}_j
 =: \hat{\bm{\sigma}}_j
 \label{eq:staggered_spin_rotation}
\end{equation}
The Hamiltonian is transformed by $\hat{R}(\phi)$ as follows:
\begin{align}
 \hat{R}(\phi)^{-1} \hat{\mathcal{H}}_{\mathrm{map}} \hat{R}(\phi)
 &= \tilde{J} \sum_{j=1}^{L} (\hat{\sigma}_j^x \hat{\sigma}_{j+1}^x + \Delta \hat{\sigma}_j^y \hat{\sigma}_{j+1}^y + \hat{\sigma}_j^z \hat{\sigma}_{j+1}^z) - \tilde{H} \sum_{j=1}^{L} \hat{\sigma}_j^z \notag \\
 &= \tilde{J} \sum_{j=1}^{L} \left[ \Delta \hat{\bm{s}}_j \cdot \hat{\bm{s}}_{j+1} + (-1)^j \sqrt{1 - \Delta^2} (\hat{s}_j^x \hat{s}_{j+1}^z - \hat{s}_j^z \hat{s}_{j+1}^x) \right]
 - \frac{\tilde{H}}{\sqrt{2}} \sum_{j=1}^{L} \left[ \sqrt{1 + \Delta} \, \hat{s}_j^z + (-1)^j \sqrt{1 - \Delta} \, \hat{s}_j^x \right].
\end{align}
On the classical line $\tilde{H} = \tilde{H}_{\mathrm{cl}}$, in particular, this equation reduces to
\begin{equation}
 \hat{R}(\phi)^{-1} \hat{\mathcal{H}}_{\mathrm{map}} \hat{R}(\phi)
 = \tilde{J} \sum_{j=1}^{L} \left[ \Delta \hat{\bm{s}}_j \cdot \hat{\bm{s}}_{j+1} - (1 + \Delta) \hat{s}_j^z + (-1)^j \sqrt{1 - \Delta^2} (\hat{s}_j^x \hat{s}_{j+1}^z - \hat{s}_j^z \hat{s}_{j+1}^x - \hat{s}_j^x) \right].
\end{equation}
Then the following statement holds:
\begin{proposition}
 \label{prop:RHR_GS}
 Let $\ket{0}$ be a fully polarized state $\ket{\uparrow \uparrow \dots \uparrow}$.
 Then $\ket{0}$ is the ground state of $\hat{R}(\phi)^{-1} \hat{\mathcal{H}}_{\mathrm{map}} \hat{R}(\phi)$ with its eigenenergy being $E_0 = - \tilde{J} L \left(\frac{1}{2} + \frac{\Delta}{4}\right)$.
\end{proposition}
\begin{proof}
 One can easily confirm that $\ket{0}$ is an eigenstate of $\hat{R}(\phi)^{-1} \hat{\mathcal{H}}_{\mathrm{map}} \hat{R}(\phi)$ and its eigenenergy is equal to $E_0$.
 We thus show that $E_0$ is the lowest energy.
 First, let $[\hat{s}_j^{\bar{x}}, \hat{s}_j^{\bar{y}}, \hat{s}_j^{\bar{z}}]^{\mathrm{T}}$ be $U_j(2\phi) [\hat{s}_j^{x}, \hat{s}_j^{y}, \hat{s}_j^{z}]^{\mathrm{T}}$, where $U_j(\phi)$ is defined in Eq.~\eqref{eq:staggered_spin_rotation}.
 Then the following equation can be derived after some algebra:
 \begin{align}
  & \hat{R}(\phi)^{-1} \hat{\mathcal{H}}_{\mathrm{map}} \hat{R}(\phi) + \tilde{J} L \left(\frac{1}{2} + \frac{\Delta}{4}\right) \notag \\
  &= \tilde{J} \sum_{j=1}^{L} \left\{ \Delta \left[ \frac{1}{4} - (\hat{s}_j^\xi \hat{s}_{j+1}^\xi + \hat{s}_j^\eta \hat{s}_{j+1}^\eta + \hat{s}_j^z \hat{s}_{j+1}^z) \right] + \left(\frac{1}{2} - \hat{s}_j^z\right) \left(\frac{1}{2} - \hat{s}_{j+1}^{\bar{z}}\right) + \left(\frac{1}{2} - \hat{s}_j^{\bar{z}}\right) \left(\frac{1}{2} - \hat{s}_{j+1}^z\right) \right\},
 \end{align}
 where $\hat{s}_j^{\xi, \eta} := (-1)^j \hat{s}_j^{x, y}$.
 Noting that $\hat{s}_j^\xi \hat{s}_{j+1}^\xi + \hat{s}_j^\eta \hat{s}_{j+1}^\eta + \hat{s}_j^z \hat{s}_{j+1}^z \leq 1/4$ and $\hat{s}_j^{z, \bar{z}} \leq 1/2$, we obtain
 \begin{equation}
  \hat{R}(\phi)^{-1} \hat{\mathcal{H}}_{\mathrm{map}} \hat{R}(\phi) + \tilde{J} L \left(\frac{1}{2} + \frac{\Delta}{4}\right) \geq 0,
 \end{equation}
 which means that $E_0$ is the lowest energy.
\end{proof}
From Proposition~\ref{prop:RHR_GS}, we straightforwardly find that $\hat{R}(\phi) \ket{0}$ is the ground state of $\hat{\mathcal{H}}_{\mathrm{map}}$ on the classical line $\tilde{H} = \tilde{H}_{\mathrm{cl}}$.
Due to the $\mathbb{Z}_2$ symmetry $\hat{Z}$, $\hat{R}(-\phi) \ket{0}$ is also the ground state.
These two states are related by the translation symmetry: $\hat{T} \hat{R}(\phi) \ket{0} = \hat{R}(-\phi) \ket{0}$.
Therefore, there exist doubly degenerate ground states with momenta $0$ and $\pi$ on the classical line:
\begin{equation}
 \begin{cases}
  \hat{R}(\phi) \ket{0} + \hat{R}(-\phi) \ket{0} & k = 0, \\
  \hat{R}(\phi) \ket{0} - \hat{R}(-\phi) \ket{0} & k = \pi.
 \end{cases}
\end{equation}

\section{Semiclassical approach}
In this section, we employ a slightly different narrative and reiterate the essence of the main text from the perspective of semiclassical field theory.
Our aim is to convince the reader that the underlying theme of our work as seen from this vantage point is that of \textit{spin Berry phase} effects.
It takes only a cursory survey of the current theoretical literature on magnetism to see that the common thread to much of its quantum exotica boils down to just this concept~\cite{Sachdev_2023_SM}.

In the language used below our main offering amounts to a simple recipe independent of dimensionality, which generates \textit{lattice Hamiltonians describing antiferromagnets whose low-energy physics are governed by spin Berry phases}.
(The precise way in which the influence of Berry phases manifests itself will depend on details, such as perturbative terms added onto our basic Hamiltonian.)
The scheme only works for systems built of half-odd-integer spin moments; the low-energy sector, which potentially generates spin Berry phases, is quenched when the Hamiltonian consists of integer-valued spins.
This constitutes a reasonably generic route to designing a spin parity effect.
It is worthwhile to note, in contrast, that while all examples of known spin parity effects that we mentioned in the main text can be understood in terms of spin Berry phases, their specific mechanisms, as well as their Hamiltonians vary, are strongly tied to spatial dimensionality as well as lattice geometry, and are not always easy to construct systematically.

This portion of the Supplemental Materials consists of subsections A--D.
The following is a layout of how the discussion evolves below:
\begin{itemize}
 \item
 Section~\ref{sec:semiclassical_idea} is an exposition of our general scheme in a semiclassical language.
 We will see how it leads us to identify spin Berry phases as the factor discriminating spin systems built of half-odd-integer spins from their integer spin counterpart.
 \item
 Section~\ref{sec:higher_dimensions} spells out briefly what the implications of our strategy are when we broaden our range of spin systems to higher dimensions and more general interactions.
 \item
 Section~\ref{sec:semiclassical_Zeeman} describes an application of our general strategy to the quantum spin chain model taken up in the main text, which is an antiferromagnet with an anisotropic interaction and an applied transverse magnetic field.
 Here we arrive at an effective field theory that enables us to make contact with a recent work on the magnetization behavior of \textit{chiral ferromagnetic} spin chains~\cite{Kodama2023_SM}, which also features a spin parity effect.
 In both work the spin parity effect has its source in the spin Berry phase term of the effective action.
 \item
 Section~\ref{sec:semiclassical_Kzero} provides, based on the same semiclassical point of view, a study on how the magnetization behavior of finite-size spin chains changes when we remove the single-ion anisotropy term from our Hamiltonian.
 The characteristic pseudo-plateaus showing up in the magnetization curve (to be distinguished from magnetization plateaus which are stable in the thermodynamic limit) now ceases to be a spin parity effect as they are present for all values of $S$.
 They nevertheless \textit{can} still be understood to be a spin Berry phase effect;
 the latter is represented by an action previously employed in a semiclassical analysis of magnetization plateaus~\cite{TanakaTotsukaHu2009_SM}.
 It is interesting that both the magnetization plateau and the pseudo-plateau are, in the semiclassical framework consequences of the same Berry phase term.
 The TLL characterization of pseudo-plateaus which played an important role in the main text is given a natural interpretation that follows from our effective theory.
 \item
 Section~\ref{sec:semiclassical_summary} gives a short summary of the earlier subsections, and in particular tries to place the contents of Secs.~\ref{sec:semiclassical_Zeeman} and \ref{sec:semiclassical_Kzero} on a common footing.
\end{itemize}

\subsection{The main idea}
\label{sec:semiclassical_idea}
Consider the Hamiltonian studied in the main text:
\begin{equation}
 \mathcal{H} = \sum_{j} \left[ J \bm{S}_j \cdot \bm{S}_{j+1} + K (S^z_j)^2 - H S^x_j \right],
 \label{eq:1d_Hamiltonian}
\end{equation}
and take the limit where the $K$-term well dominates over the exchange and Zeeman interactions, i.e., $K \gg J, H$.
The leading contribution to the Hamiltonian, the $K$-term, acts locally. Let us view this situation within a semiclassical approach, where we employ the machinery of the spin coherent state path integral.
Denoting the spin vector as $\bm{S}(\tau, x) = S\bm{n}(\tau, x)$, where $\bm{n}(\tau, x)$ has unit norm, the decoupled action which arises at each site in this ``atomic'' limit consists of two local terms: the spin Berry phase (or Wess--Zumino) term, and a rotor-like inertial term, which can be expressed in Euclidean spacetime as
\begin{equation}
 \mathcal{S}_\mathrm{atomic} = -\ii S \omega[\bm{n}(\tau)] + \int \mathrm{d}\tau \, \frac{\chi}{2} (\partial_{\tau}\bm{n})^2.
 \label{eq:rotor_action}
\end{equation}
The quantity $\omega[\bm{n}(\tau)]$ is the surface area on the unit sphere swept out by the imaginary-time dynamics of the rotor $\bm{n}(\tau)$.
Though the $\chi$-term was introduced above by hand, we will later see that it arises dynamically from our Hamiltonian in the course of integrating out the higher-energy degrees of freedoms.
Were we to ignore in \eqref{eq:rotor_action} the effect of anisotropy, this would precisely be the action of a point particle of unit electric charge, constrained to move on the surface of a unit sphere with a Dirac magnetic monopole of strength $2S$ sitting at its center~\cite{fradkin_book_2013_SM}.
The eigenstates are monopole harmonics~\cite{Wu_Yang_1976_SM}, which exhibit a $(2S+1)$-fold degeneracy in the ground state.

We will now argue that the anisotropic nature of the $K$-term can bring an altogether different picture into the low-energy physics.
The basic assumption is that the latter term effectively confines the motion of the rotor $\bm{n}$ to the $xy$ plane.
To see its consequences, we represent with $\theta$ the angular orientation of $\bm{n}$ within this plane.
The planar limit of action \eqref{eq:rotor_action} then reads
\begin{equation}
 \mathcal{S}_\mathrm{atomic}
 \stackrel{\text{planar}}{=}
 \int \mathrm{d}\tau \mathrm{d}x \left[-\ii S \partial_{\tau}\theta + \frac{\chi}{2} (\partial_{\tau}\theta)^2\right].
 \label{eq:planar_atomic_limit_action}
\end{equation}

\begin{figure}[h]
 \includegraphics[width=.5\textwidth]{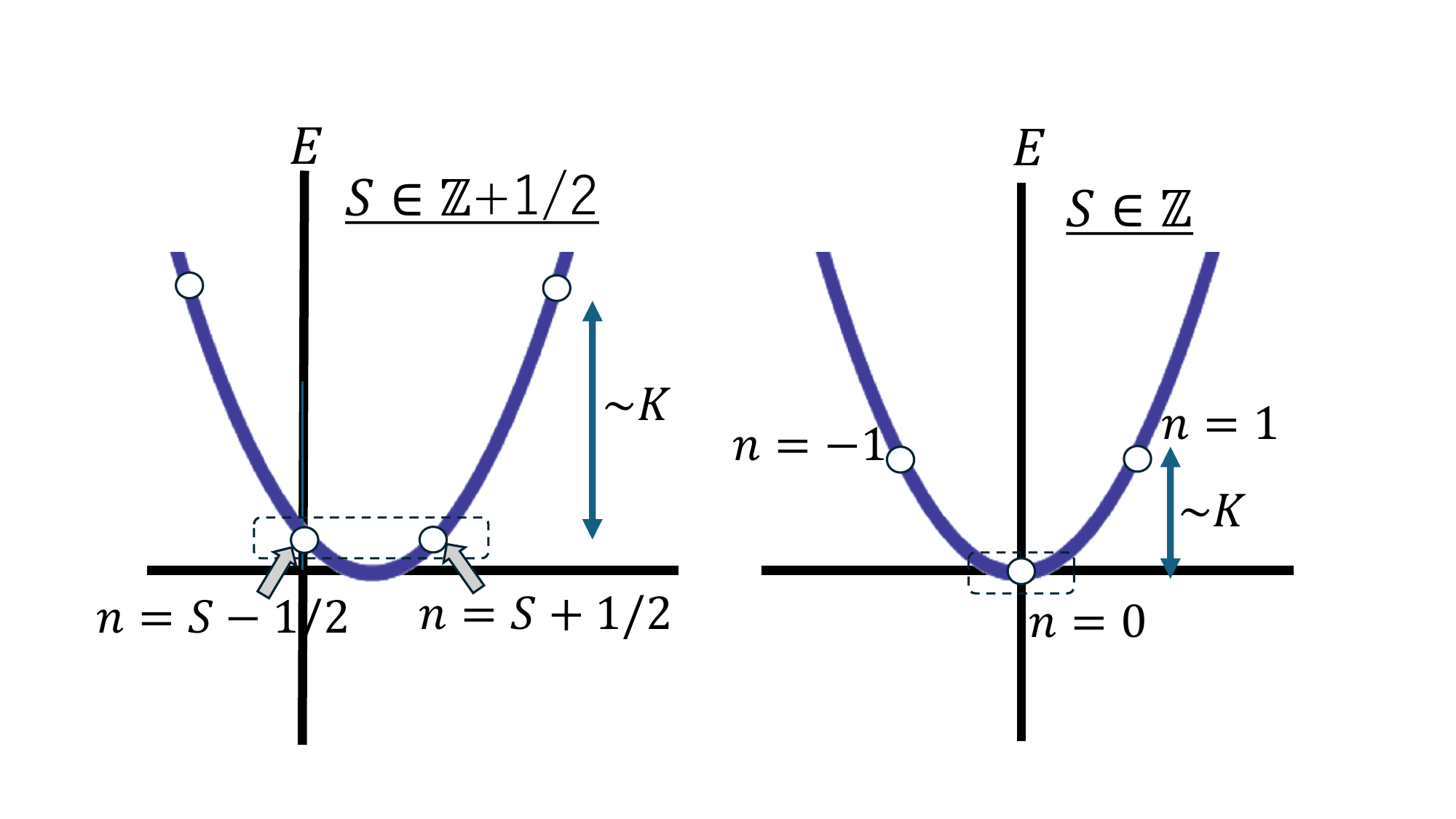}
 \caption{Low energy portion of the energy eigenvalues of $\hat{\mathcal{H}}_\mathrm{atomic}$ for the half-odd-integer $S$ (left) and integer $S$ (right) cases. The lowest energy sector is spanned by states within the dotted squares.}
 \label{fig:low_energy_levels}
\end{figure}

\noindent
Figure~\ref{fig:low_energy_levels} depicts the low-energy part of the eigenstates of the Hamiltonian corresponding to the action \eqref{eq:planar_atomic_limit_action},
\begin{equation}
 \hat{\mathcal{H}}_\mathrm{atomic} = \frac{1}{2\chi} \left(\hat{\pi}_{\theta} - S\right)^2 ,
\end{equation}
where $\hat{\pi}_{\theta} := -\ii \frac{d}{d\theta}$ is the momentum conjugate to $\theta$.
Noting that the eigenvalue $n$ of $\hat{\pi}_{\theta}$ is integral valued, it is clear from the figure that the lower-most sector of the spectrum, separated from higher energy states by an energy gap of order $K$, is a doublet when $S$ is a half-odd integer, while being nondegenerate when $S$ is integral.

For the remainder of this subsection, we will set the magnetic field to zero, $H = 0$;
the effect of the Zeeman term shall be the subject of Sec.~\ref{sec:semiclassical_Zeeman}.
Adding on back the exchange interaction (here assumed to be secondary, i.e., that its energy scale lies within the gap $\sim K$), we now have planar spins sitting on each site which interact with its nearest neighbors in an antiferromagnetic fashion.
To better describe this situation, it is now more convenient to have $\bm{n}$ and $\theta$ represent the orientation of the \textit{staggered} moments, i.e.,
\begin{equation}
 \bm{S}_j = (-1)^j S \bm{n}_j, \quad
 \bm{n}_j = (\cos\theta_j, \sin\theta_j, 0).
 \label{eq:staggered_moments}
\end{equation}
(This redefining procedure will be dealt with more systematically in Sec.~\ref{sec:semiclassical_Zeeman}.)
The key observation to make is that \textit{the spin Berry phase term behaves as if $\bm{n}$ represents an $S = 1/2$ moment when $S$ is half-odd integer-valued, while this contribution is totally absent when $S$ is an integer}.
Feeding \eqref{eq:staggered_moments} into \eqref{eq:planar_atomic_limit_action} and switching to a continuum description~\cite{Haldane_PRL_1988_SM}, we obtain (we shown in more detail later) a quantum XY model bearing the form
\begin{equation}
 \mathcal{S}[\theta(\tau, x)] = \int \mathrm{d}\tau \mathrm{d}x \left[ -\ii \frac{S_\mathrm{eff}}{a} \partial_{\tau}\theta + \frac{\chi}{2} (\partial_{\tau}\theta)^2 + \frac{\kappa}{2} (\partial_x \theta)^2 \right],
\end{equation}
where $a$ is the lattice constant and $S_\mathrm{eff}$ is either $1/2$ or $0$ as explained above.
(An explanation is in order as to what the above equation describes for the integer $S$ case, as it superficially appears to represent a massless Tomonaga--Luttinger liquid.
On the contrary, the lack of intervention of Berry phase factors dictates that the system is typically unstable toward gap-opening through the effect of vortices.
Formally the latter are vertex operators of the field dual to $\theta$, and can be generated as they do not break the inherent U(1) symmetry.
We will come back to this point shortly when discussing the $S = 1$ case.)
We note that the nonuniversal coefficients $\chi$ and $\kappa$ are subject to renormalization effects, while the precise value of the effective spin $S_\mathrm{eff}$ governs the pattern of phase interference among the space-time configurations entering the partition function $Z = \int \mathcal{D}\theta(\tau, x) e^{-\mathcal{S}[\theta(\tau, x)]}$, as we will now see.

There are several different ways to go about this task;
we will discuss two, each leading to the conclusion that the Berry phase term acts to render the system ingappable when $S$ is a half-odd integer.
The first argument borrows from duality arguments detailed in Ref.~\cite{TanakaTotsukaHu2009_SM};
noticing that spacetime vortices are the prime source of a potential mass-gap generation, we re-express the action in terms of the vortex density $\rho_\mathrm{v} := \frac{1}{2\pi} \epsilon_{\mu\nu} \partial_{\mu} \partial_{\nu}\theta$ ($\mu = \tau, x$) to find that it has the form $\mathcal{S}_\mathrm{v} = \mathcal{S}_\mathrm{Coulomb} + \ii 2\pi S_\mathrm{eff} \int \mathrm{d}\tau \mathrm{d}x \rho_\mathrm{v}x$, where the first term is a two-body logarithmic interaction between vortices and the second a Berry phase term.
Restricting as usual to vortices with the lowest vorticities $\pm 1$, we consider the effect of introducing a test unit vortex into the system.
One can read off that when $S_\mathrm{eff} = 1/2$, two spacetime configurations shifted from each other in the $x$-direction by a lattice constant (but is otherwise identical) differ in their Feynman weight $e^{-\mathcal{S}_\mathrm{v}}$ by a Berry phase factor of $-1$.
The pairwise cancellation suggests that for half-odd-integer $S$, vortices are always confined and fail to disorder the system.
For the second argument, we go back to the lattice and note that for any given site, flipping the sign of the spin Berry phase term is immaterial for any $S$ since the difference in phase is always an integer multiple of $2\pi$.
We take the continuum limit after this sign flip has been performed for every other site.
A small fugacity expansion into a gas of vortices, again with the lowest vorticities, results in a sine-Gordon model.
The spin Berry phase fixes the prefactor of the cosine term; it is found to be proportional to $\cos(\pi S_\mathrm{eff})$~\cite{Affleck_merons_PRL_1986_SM}, implying once again that for the half-odd-integer $S$ case, a mass gap cannot open up even when vortices are relevant in the renormalization group sense.

Despite it being a null result, the implications for integer $S$ provide a useful sanity check for our scheme.
Consider the case $S = 1$.
Assuming that our picture is valid, the fugacity expansion yields a cosine term with the prefactor $\cos(\pi S_\mathrm{eff}) \stackrel{\text{$S_\mathrm{eff} = 0$}}{=} 1$.
This is to be compared with the corresponding prefactor $\cos(\pi S) \stackrel{\text{$S = 1$}}{=} -1$, which arises if the same expansion were carried out for an $S = 1$ spin chain in the Haldane gap phase with an easy-plane anisotropy~\cite{Affleck_merons_PRL_1986_SM}.
The sign difference between the two cases is consistent with the fact that for $S = 1$, the large-$K$ phase (where the ground state is a trivial product state) is distinct from the Haldane gap state (a typical symmetry-protected topological phase), with a gap-closing point intervening between the two~\cite{Gu_Wen_PRB_2009_SM, Pollmann_PRB_2012_SM}.
These observations also agree with abelian bosonization studies~\cite{Schulz_PRB_1986_SM, Fuji_PRL_2009_SM}.

\subsection{Higher dimensional generalization and beyond}
\label{sec:higher_dimensions}
As noted in the main text, the ``atomic-limit'' picture depicted in Fig.~\ref{fig:low_energy_levels} has a range of validity that goes far beyond the specific model of Eq.~\eqref{eq:1d_Hamiltonian}.
It is in fact generally correct regardless of dimensionality, lattice geometry \textit{or} the Hamiltonian itself, provided other interactions are sufficiently small in their magnitudes compared to the single-site anisotropy. We discuss what this observation implies for higher dimensional antiferromagnets.
For integer $S$, there apparently is not much difference in the outcome depending on the specifics of the system under consideration: we generally expect a featureless, gapped ground state.
Our main concern therefore will be on the rich variety of quantum effects that can arise in the half-odd integer $S$ cases.

Before delving into 2D cases, it is instructive to gain intuition using spin chains with half-odd integer $S$.
We saw in the previous subsection that for this case, unit vortices interfere destructively.
This argument also tells us that this interference can be removed if vortices are paired, i.e. have vorticity $\pm 2$, hence leading to vortex condensation and a gap.
It is well known~\cite{Haldane1982_PRB_SM, *[][{ (erratum).}]{Haldane1982_PRB_erratum_SM}, Park_Sachdev_AnnPhys_2002_SM} that for $S = 1/2$, such doubly quantized vortices are stabilized by a second nearest neighbor antiferromagnetic exchange, and will result---when the ratio $J_2 / J_1$ between the next-nearest neighbor ($J_2$) and nearest neighbor ($J_1$) exchange integrals is sufficiently large, in a spontaneously dimerized ground state with a two-fold degeneracy~\cite{Haldane1982_PRB_SM, Read1989_SM,Park_Sachdev_AnnPhys_2002_SM}.
A solvable realization of this scenario is the $S = 1/2$ Mazumdar--Ghosh model with $J_2 / J_1 = 1/2$.
Using the perturbation scheme of the main text, we readily see that a higher $S$ spin chain (where $S$ is half-odd integral) consisting of the $K$-term as well as the $J_1$- and $J_2$-terms can be mapped, in the dominant-$K$ regime, into a $S = 1/2$ XYX model with $J_1$ and $J_2$ interactions.
Summarizing, the dichotomy that we derive for this problem takes the form:
\begin{equation}
 \begin{array}{ccc} \hline\hline
  & \text{integer $S$} & \text{half-odd integer $S$} \\ \hline
  \text{small $J_2 / J_1$} & \text{trivial gapped state} & \text{TL state} \\
  \text{large $J_2 / J_1$} & \text{trivial gapped state} & \text{twofold degenerate spin-Peierls state} \\ \hline\hline
 \end{array}
 \label{eq:dichotomy}
\end{equation}
That these two forms of spin parity effects should arise is in complete accord with the LSM theorem as applied to 1D spin systems.
We can also ask ourselves if a distinction in spin Berry phases arises between the small- and large-$K$ regimes.
The answer, while subtle, turns out to be affirmative.
(The distinction becomes somewhat clearer in higher dimensions.)
The main difference manifests itself in the integer $S$ case.
While the large-$K$ ground state is a featureless product state void of Berry phase effects, the effective action at small $K$ does contain a spin Berry phase term, which generates nontrivial surface contributions when $S$ is odd.
(A related observation regarding the difference in vortex Berry phases for $S = 1$ in the two regimes was also mentioned in the previous subsection.)
The surface term signals the presence of boundary states which appear under an open boundary condition.
It also reflects the sensitivity of the ground state wave function~\cite{Takayoshi_Totsuka_Tanaka2015_SM} to the
global topology of the snapshot configuration.

One can proceed in a similar fashion in 2D. For concreteness let us consider antiferromagnets on a square lattice, first without the $K$-term.
The relevant singular space-time events, essentially playing the same role as the vortices in one dimension lower, are \textit{monopoles}.
(A unit monopole tunnels the system between configurations whose Skyrmion numbers differ by one.)
One can show~\cite{Haldane_PRL_1988_SM,Read1989_SM, Park_Sachdev_AnnPhys_2002_SM} based on the associated Berry phase factors, that in order for monopoles to be able to condense and create an energy gap, they will have to bundle up into quadruples when $S = 1/2 \pmod{2}$ or $S = 3/2 \pmod{2}$, and form pairs when $S = 1 \pmod{2}$.
When $S = 2 \pmod{2}$ they can condense without any form of grouping.
Correspondingly, the ground state degeneracy when monopoles become relevant depends on $S$ modulo $2$, and is four (when $S = 1/2$ or $S = 3/2 \bmod 2$), two (when $S = 1 \bmod 2$), or one (when $S = 2 \bmod 2$).
The degenerate states are either spin Peierls-like or plaquette dimer states, which can be shifted by one site or rotated by $\pi/2$ to transform into their degenerate counterparts.
As in the spin chain case, the monopoles events are expected to become relevant with the addition of frustration, as well as some appropriate four-spin terms.
With this information at hand, we now introduce the $K$-term.
The integer $S$ case will once again yield trivially gapped states, while for the half-integer $S$ case, the ground state can be expected to be gapless if monopoles are irrelevant, and (reflecting the aforementioned $S = 1/2$ result at $K = 0$) four-fold degenerate when they are relevant.
We can again see that the Berry phase effects are clearly different between cases where $K$ is small and large.
(This is most easily seen by comparing the two regimes for the $S = 1$ case.)
This analysis can be repeated for the case when the spins reside on a honeycomb lattice, which leads to a different pattern (dependent on $S$ modulo $3/2$) of grouping of the monopoles.

We conclude this subsection by adding that the ``atomic limit'' picture can be applied as well to nonbipartite (e.g. triangular) lattices.
When $S = 1/2$, this may lead on general grounds to a topologically ordered ground state with fractionalized, anyonic excitations, which is the third and last alternative that the generalized LSM theorem leaves room for.
While requiring further investigation, large $S$ spin system simulated by cold-atom systems, when projected onto an effective $S = 1/2$ system by incorporating large-$K$ interactions, may lead, along this line to exotic ground states.

\subsection{Adding on the Zeeman interaction}
\label{sec:semiclassical_Zeeman}
Coming back to the spin chain model of the main text, Eq.~\eqref{eq:1d_Hamiltonian}, we now provide a semiclassical description for the $H \neq 0$ case based on the large-$S$ mapping of an antiferromagnetic spin chain to the O(3) nonlinear sigma model~\cite{Haldane_PRL_1988_SM}.
Here, as well as in Sec.~\ref{sec:semiclassical_Kzero}, the spatial extent of our spin chain will be denoted as $L := N_\mathrm{site} \times a$, with the number of sites $N_\mathrm{site}$ chosen to be an even integer.
(Please note that this notation differs from that used in the main text: there, $L$ stood for the integer-valued site number, and the lattice constant was set to unity. In this section, where we deal with continuum field theories, it turns out to be convenient to retain the explicit dependence on the lattice constant. Hence the use of the notation $a$.)
Periodic boundary conditions will always be assumed.
A convenient point of departure is the corresponding action for the case \textit{without} the $K$-term.
This reads:
\begin{equation}
 \mathcal{S}[\bm{n}(\tau, x)] = \mathcal{S}_\mathrm{top}
 + \int \mathrm{d}\tau \mathrm{d}x \left\{ \frac{1}{4Ja} \left[(\partial_{\tau}\bm{n} - \ii \bm{H} \times \bm{n})^2 - (\bm{H} \cdot \bm{n})^2 \right] + \frac{JS^2 a}{2} (\partial_x \bm{n})^2 \right\}.
\end{equation}
Let us consider how the presence of the $K$-term will modify this action.
We start with the second term on the right-hand side.
The piece involving the vector product $\bm{H} \times \bm{n}$ will be relevant if the precessing of the staggered magnetization around the magnetic field is appreciable.
One expects, however, that it will be largely suppressed by the anisotropy term, and the action will be dominated by the dynamics of the field $\theta$, which we encountered in Sec.~\ref{sec:semiclassical_idea}.
We next turn to the topological term $\mathcal{S}_\mathrm{top}$, which in the isotropic case is well known to be $\mathcal{S}_\mathrm{top} = -\ii 2\pi S Q_{\tau x}$, where $Q_{\tau x}$ is the integer-valued Skyrmion number for $\bm{n}(\tau, x)$.
In the planar situation, which we are interested in, there are two alternative ways to modify this action, which are basically equivalent~\cite{Takayoshi_Totsuka_Tanaka2015_SM}.
One is to continue using the same form while incorporating meron configurations, for which the Skyrmion numbers are halved, i.e., $Q_{\tau x} = \pm 1/2$ in the least costly configurations~\cite{Affleck_merons_PRL_1986_SM}.
The other is to recall the discussion of Sec.~\ref{sec:semiclassical_idea} and undo the staggering of the spin Berry phases at each site prior to taking the continuum limit.
This will yield as before the action $\mathcal{S}_\mathrm{top} = \ii \frac{S_\mathrm{eff}}{a} \int \mathrm{d}\tau \mathrm{d}x \partial_{\tau}\theta(\tau, x)$.
Adopting the second form, we arrive at a quantum sine-Gordon model:
\begin{equation}
 \mathcal{S}[\theta] = \int \mathrm{d}\tau \mathrm{d}x
 \left\{ -\ii \frac{S_\mathrm{eff}}{a} \partial_{\tau}\theta + \frac{1}{2g} \left[\frac{1}{v^2} (\partial_{\tau}\theta)^2 + (\partial_x \theta)^2\right] - \frac{1}{4Ja}H^2 \cos^2\theta \right\},
 \label{eq:quantum_SG_action}
\end{equation}
where $g = \frac{1}{JS^2 a}$ and $v = \sqrt{2} aSJ$.
Once formulated in this way, we can follow closely the analysis carried out in the appendix of Ref.~\cite{Kodama2023_SM} for the case of chiral ferromagnets, to see how the sine-Gordon solitons and their Berry phases
affect the magnetization process.
We first note that for the present problem, the Euler--Lagrange equation for $\theta$ admits a $\pi$-soliton as the static solution with the lowest domain wall height.
This solution however is a spinon, across which the two sublattices interchange roles, and is thus in conflict with the periodic boundary condition.
We should, therefore, focus on the effect of $2\pi$ solitons.
(The admission of $\pi$-solitons into the low-energy theory will lead to an interplay between spin parity and soliton chirality~\cite{Braun_Loss_PRB_1996_SM}.)
Let us represent such a solution for our sine-Gordon equation $\partial_{xx}\theta = \frac{M}{2} \sin 2\theta$ as $\theta(x) := \Theta_{2\pi}(x - X)$, where $M := \frac{H}{\sqrt{2}JSa}$ and $X$ is the center coordinate (hereafter called the collective coordinate) of the soliton.
By definition $\int \mathrm{d}x \partial_x \Theta_{2\pi} = 2\pi$, and from a virial-theorem-like argument one also finds that $\int \mathrm{d}x (\partial_x \Theta_{2\pi})^2 = M^2$.
Having identified the relevant static solution, we now promote the collective coordinate to a dynamical degree of freedom $X(\tau)$.
Substituting $\theta = \Theta_{2\pi}(x - X(\tau))$ into Eq.~\eqref{eq:quantum_SG_action} we obtain an effective action for $X(\tau)$:
\begin{equation}
 \mathcal{S}[X(\tau)] = \int \mathrm{d}\tau \left[ \ii \frac{2\pi S_\mathrm{eff}}{a} \partial_{\tau}X + \frac{1}{2} \tilde{M}(\partial_{\tau}X)^2 \right],
\end{equation}
where $\tilde{M} := \frac{H^2}{4(Ja)^3 S^2}$.
The corresponding Hamiltonian is:
\begin{equation}
\hat{H}_{X} = \frac{1}{\tilde{M}} \left(\hat{\pi}_X - \frac{2\pi S_\mathrm{eff}}{a}\right)^2 + V(x),
\end{equation}
where $\hat{\pi}_X := -\ii \frac{d}{dX}$ is the momentum operator conjugate to $X$, and we have added a potential energy $V(X)$ with the periodicity $V(x + a) = V(x)$ as a remnant of the underlying lattice.
The energy band experienced by a single minimal soliton in motion under the influence of the potential energy $V(x)$ has a minimum at the wavenumber $k = \frac{2\pi S_\mathrm{eff}}{a}$.
The above steps can be repeated for higher domain walls, with an increment in $k$ by $\pi$ for the case $S_\mathrm{eff} = 1/2$ for each additional winding of $\phi(x)$ along the extent of the spin chain.
As the $2\pi$ solitons for the case $S_\mathrm{eff} = 1/2$ are magnons, the system can sustain up to $N_\mathrm{sites}/2$ such objects provided this simple picture of domain-wall generation holds up all the way, at least in a qualitative sense, through the magnetizing process.

\subsection{$K = 0$}
\label{sec:semiclassical_Kzero}
The gapless phase of the antiferromagnetic spin chain under a magnetic field at $K = 0$ is generally believed to be a variant of the U(1)-symmetric Tomonaga--Luttinger liquid (TLL).
The universal structures of TLLs, which played a central role in the main text, are strongly constrained by requirements imposed by conformal field theory.
An effective field theory of this problem should share these structures and, moreover, provide an explicit derivation of the relevant quantum numbers that characterize the TLL description of our spin chain.
A possible route to derive such a theory is to use bosonization methods.
There exist several versions of this technique that can be applied to spin chains with arbitrary $S$, such as that developed by Schulz~\cite{Schulz_PRB_1986_SM}.
In the same spirit as in preceding subsections though, we will take a simple semiclassical perspective and see how it enables us to connect to and reinforce the discussions of the main text.
Once again, keeping tract of the spin Berry phase will be seen to be essential, although now their implications will depend not on the spin parity, but on the quantity $S - m$, where $m$ is the magnetization density.

In this subsection, we will choose to align the applied magnetic field with the $z$-axis.
Taking advantage of the U(1) symmetry within the $xy$ plane we employ a ``canted ansatz'' for partially magnetized spin moments~\cite{TanakaTotsukaHu2009_SM}, $\bm{S}_j = S\bm{N}_j$, where
\begin{equation}
 \bm{N}_j = \left((-1)^j\sqrt{S^2-m^2}\cos\phi_j, (-1)^j\sqrt{S^2-m^2}\sin\phi_j, m \right).
 \label{eq:canted_spin_ansatz}
\end{equation}
The phases $\phi_j$ will be treated as the slowly varying degree of freedom.
\begin{figure}[h]
 \includegraphics[width=0.3\textwidth]{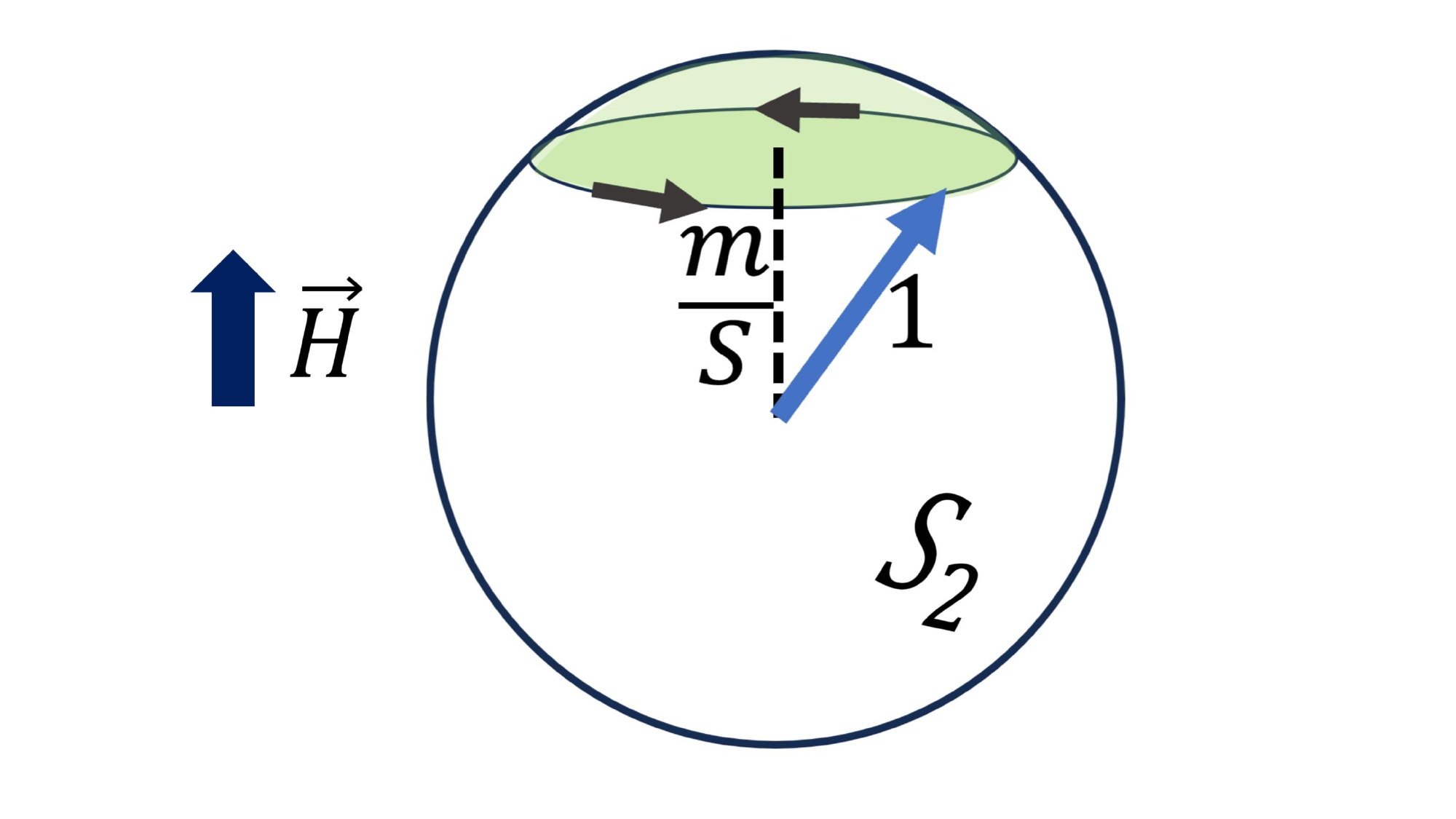}
 \caption{Precessing motion of a canted spin moment projected onto the unit sphere $S_2$. The surface area (shaded region), multiplied by $S$ gives the spin Berry phase.}
 \label{fig:canted_spin}
\end{figure}
It is easily verified that (see Fig.~\ref{fig:canted_spin}) the Berry phase action of a spin residing at site $j$ is
\begin{equation}
 \mathcal{S}_\mathrm{BP}[\phi_j(\tau)] = -\ii (S - m) \int \mathrm{d}\tau \partial_{\tau}\phi_j(\tau).
 \label{eq:BP_term_for_a_canted_spin}
\end{equation}
The \textit{non-staggered} nature of the action \eqref{eq:BP_term_for_a_canted_spin} owes to the fact that only two components of the spin vector \eqref{eq:canted_spin_ansatz} have the staggered factors $(-1)^j$.
Substituting the ansatz \eqref{eq:canted_spin_ansatz} into the lattice Hamiltonian with $K = 0$ and taking the continuum limit, we find that the basic form of the effective action for our TLL is written in terms of the angular field $\phi(\tau, x)$ and the magnetization density $m(\tau, x)$ as:
\begin{equation}
 \mathcal{S}[m, \phi] = \int \mathrm{d}\tau \mathrm{d}x
 \left[ -\ii \frac{S - m}{a} \partial_\tau \phi + Jm^2 + \frac{Ja^2}{2} (S^2 - m^2) (\partial_x \phi)^2 - mH \right].
 \label{eq:effective_TL_action}
\end{equation}
Allowing for fluctuations of the magnetization around a uniform value, $m = m_0 + \delta m(\tau, x)$, and integrating over $\delta m(\tau, x)$, the leading terms of our action are found to be
\begin{equation}
 \mathcal{S}[\phi] = \int \mathrm{d}\tau \mathrm{d}x
 \left[ -\ii \frac{S - m_0}{a} \partial_\tau \phi + \frac{1}{4Ja^2} (\partial_\tau \phi)^2 + \frac{Ja^2}{2} (S^2 - m_0^2) (\partial_x \phi)^2 \right],
 \label{eq:leading_order_action}
\end{equation}
where we have omitted the Zeeman term, which is a source term for $m$.
This is a TLL action with a Berry phase term.
Within this description, the latter term controls the gappability of this system~\cite{TanakaTotsukaHu2009_SM, Takayoshi_Totsuka_Tanaka2015_SM}.
As our aim is in investigating the massless TLL phase, we will assume that the quantization conditions~\cite{Oshikawa_Yamanaka_Affleck_SM} which stabilizes a genuine magnetization plateau (i.e., those that persist in the large $L$ limit, to be distinguished from the pseudo-plateaus~\cite{Parkinson1985_SM} discussed in the main text) are not met.

The Hamiltonian corresponding to \eqref{eq:leading_order_action} reads
\begin{equation}
 \mathcal{H}[\hat{\pi}_{\phi}, \phi] = \int \mathrm{d}x
 \left[ Ja^2 \left(\hat{\pi}_{\phi} - \frac{S - m_0}{a}\right)^2 + \frac{Ja^2}{2}(S^2 - m_0^2) (\partial_x \phi)^2 \right].
\end{equation}
The operator $\hat{\pi}_\phi(x) := -\ii \frac{\delta}{\delta \phi(x)}$ is the momentum canonically conjugate to $\phi(x)$.
The first term on the right-hand side is formally equivalent to the kinetic energy of a particle minimally coupled to a gauge field [i.e., it is of the form $\sim (p - eA)^2$].
One can utilize this gauge degree of freedom to relate the ground state wave functional $\Psi[\phi(x)]$ to its counterpart in the absence of the gauge coupling, $\Psi_0 [\phi(x)]$:
\begin{equation}
 \Psi[\phi(x)] = e^{\ii \frac{S - m_0}{a} \int \mathrm{d}x \phi(x)} \Psi_0[\phi(x)].
\end{equation}
We will now show that the prefactor extracted above [which has its origins in the Berry phase term of the action \eqref{eq:leading_order_action}] makes an Aharonov--Bohm effect-like contribution to the crystal momentum of the state described by $\Psi[\phi(x)]$.
To this end, let us apply a one-site translation $x \rightarrow x + a$ on this state.
We concentrate on the prefactor for the moment; the information on the crystal momentum contained in $\Psi_0[\phi(x)]$ will be discussed separately.
One easily sees that the factor will experience a phase shift of the amount
\begin{equation}
 (S - m_0) \int \mathrm{d}x \partial_x \phi = 2\pi (S - m_0) \mathcal{J}
 \label{eq:phase_shift_1}
\end{equation}
where the integer $\mathcal{J} = \frac{1}{2\pi} \int \mathrm{d}x \partial_x \phi(x)$ is the winding number of the configuration $\phi(x)$.
This phase shift can be understood to be a portion of the Galilei boost that the quantum fluid receives from a phase slip event that is associated with the formation of a configuration with nonzero $\mathcal{J}$~\cite{Wen_2004_SM}.

Further contributions to the momentum of low-energy states can be read off from the structure of the action $\mathcal{S}[m, \phi]$.
Namely, one notices from an inspection of Eq.~\eqref{eq:effective_TL_action} that the quantity playing the role of a canonical momentum coupling to $\phi$ is $\frac{\delta m}{a}$.
To correctly account for the transformation that it imposes on the field $\phi(x)$, we need to take into account the short-range antiferromagnetic correlation.
Specifically, a one-site translation has the effect of interchanging the role of the two sublattices, implying that in order to make this a symmetry transformation, a $\pi$-shift needs to be introduced in addition to the usual shifting $\phi(x) \rightarrow \phi(x + a)$.
With the interpretation that $\delta m = \frac{\Delta M}{L}$, where $\Delta M$ is the increment of the total magnetization of the low-energy excitation, the expectation value of the generator of the transformation just mentioned amounts to
\begin{equation}
 \frac{\Delta M}{L} \int \mathrm{d}x (a \partial_x \phi + \pi) = \frac{\Delta M}{L} (2\pi a \mathcal{J} + \pi L).
 \label{eq:phase_shift_2}
\end{equation}
Combining Eqs.~\eqref{eq:phase_shift_1} and \eqref{eq:phase_shift_2}, and further using the notation $m_0 = \frac{M}{L}$, the crystal momentum of our interest is
\begin{equation}
 k = 2\pi \frac{LS - M}{aL} \mathcal{J} + \frac{\Delta M}{aL} (2\pi \mathcal{J}) + \pi\frac{\Delta M}{a}.
\end{equation}
To make contact with the characteristic quantum numbers of the main text, we have the correspondence $\Delta N \leftrightarrow \Delta M$, $\Delta D \leftrightarrow \mathcal{J}$ and $2\pi\rho = 2\pi\frac{LS - M}{aL}$.
This should be compared with a similar expression for the case in the absence of an external magnetic field~\cite{Haldane1981_preprint_SM}, with the main difference being that the boson density $\rho$ in the latter case takes a value commensurate with the underlying lattice.
For magnon excitations with $\Delta M = 1$ (here we should also set $\mathcal{J} = 0$) we have $N_\mathrm{c} = LS$.

\subsection{Summary---comparing $K = 0$ and $K \gg J$}
\label{sec:semiclassical_summary}
To repeat an earlier statement, it often happens that in a semiclassical description of quantum magnetism, the
spin Berry phase becomes solely responsible for generating the relevant quantum effects.
This is true for the effective theories discussed in Secs.~\ref{sec:semiclassical_Zeeman} and \ref{sec:semiclassical_Kzero}.
For both $K \gg J$ and $K = 0$, the system was described in terms of planar (phase) fields (both of which we will denote here for convenience as $\phi$), with a low-energy action containing a spin Berry phase term of the form
\begin{equation}
 \mathcal{S}_\mathrm{BP} = \ii \int \mathrm{d}\tau \mathrm{d}x \frac{S_\mathrm{eff}}{a} \partial_{\tau}\phi.
\end{equation}
The quantity $S_\mathrm{eff}$ is the size of the ``active'' spin moment undergoing quantum dynamics.
For the case $K \gg J$, we had $S_\mathrm{eff} = 1/2$ for half-odd-integer spin systems, while $S_\mathrm{eff} = 0$ for integer spin systems, which led to the spin parity effect taken up in the main text.
Meanwhile, for the $K = 0$ case, we saw that $S_\mathrm{eff} = S - m$, where $m$ is a function of the applied magnetic field $H$.
For a given value of $S$, we are, therefore, sweeping the size of the effective spin throughout the magnetizing process, which explains the absence of spin parity effects.
[An interesting possibility remains, though, that a variant of the spin parity effect may arise for the ``spin'' $S - m$~\cite{TanakaTotsukaHu2009_SM} when $m$ takes special values.
If we exclude this possibility (whose manifestation will be affected by subtle short-range physics), the finite-size cross-over phenomenon discussed in Sec.~\ref{sec:semiclassical_Kzero} will be the only relevant quantum effect that derives directly from the Berry phase term.]

\end{document}